\documentclass[11pt, a4paper]{article}
\usepackage[margin=1in, centering]{geometry}
\usepackage{mathtools}
\usepackage{amsthm}
\usepackage{thmtools}
\makeatletter
\@ifundefined{newcounteralias}{}{\renewcommand\thmt@autorefsetup{\@xa\def\csname\thmt@envname autorefname\@xa\endcsname\@xa{\thmt@thmname}}}
\makeatother
\usepackage{dsfont}
\usepackage{amssymb}
\usepackage{biolinum}
\usepackage{mathpazo, tgpagella}
\DeclareMathAlphabet{\mathsf}{\encodingdefault}{\sfdefault}{m}{n}
\SetMathAlphabet{\mathsf}{bold}{\encodingdefault}{\sfdefault}{b}{n}
\usepackage[table]{xcolor}
\usepackage[dvipsnames]{xcolor}
\usepackage{subcaption}
\usepackage[normalem]{ulem}
\definecolor{HANADA}{RGB}{0, 98, 132}
\definecolor{KURENAI}{RGB}{203, 27, 69}
\usepackage[
    pdfstartview=FitH,
    pdfpagemode=UseNone,
    bookmarksdepth=2,
    colorlinks=true,
    citecolor=KURENAI,
    linkcolor=HANADA,
    backref=page,
    linktocpage=true
]{hyperref}
\usepackage[capitalise, nameinlink]{cleveref}
\usepackage[nottoc]{tocbibind}
\usepackage{appendix}
\usepackage{array}
\usepackage{braket}
\usepackage{ytableau}
\usepackage{graphicx}
\usepackage{enumitem}
\usepackage{thm-restate}
\usepackage{tikz}
\usepackage{tikz-cd}
\usetikzlibrary{calc}
\usepackage[framemethod=TikZ]{mdframed}
\usepackage{setspace}
\usepackage{etoolbox}
\usepackage{authblk}

\renewcommand\Affilfont{\small}
\makeatletter
\renewcommand\AB@affilsepx{\protect\\[0.1em]\protect\Affilfont}
\makeatother

\newtheorem{theorem}{Theorem}[section]

\newtheorem{lemma}[theorem]{Lemma}
\newtheorem{corollary}[theorem]{Corollary}
\newtheorem{claim}[theorem]{Claim}

\newtheorem{definition}[theorem]{Definition}

\crefname{appendix}{Appendix}{Appendices}
\Crefname{appendix}{Appendix}{Appendices}
\crefname{claim}{Claim}{Claims}
\Crefname{claim}{Claim}{Claims}
\crefname{fact}{Fact}{Facts}
\Crefname{fact}{Fact}{Facts}

\newtheoremstyle{obsstyle}
  {6pt}
  {6pt}
  {\normalfont}
  {}
  {\bfseries}
  {.}
  {.5em}
  {}

\theoremstyle{obsstyle}



\newcommand{\fig}[1]{\hyperref[fig:#1]{Figure~\ref*{fig:#1}}}
\newcommand{\eq}[1]{\hyperref[eq:#1]{(\ref*{eq:#1})}}
\newcommand{\lem}[1]{\hyperref[lem:#1]{Lemma~\ref*{lem:#1}}}
\newcommand{\thm}[1]{\hyperref[thm:#1]{Theorem~\ref*{thm:#1}}}
\newcommand{\defi}[1]{\hyperref[def:#1]{Definition~\ref*{def:#1}}}
\newcommand{\app}[1]{\hyperref[app:#1]{Appendix~\ref*{app:#1}}}
\newcommand{\fct}[1]{\hyperref[fact:#1]{Fact~\ref*{fact:#1}}}
\newcommand{\clr}[1]{\hyperref[clr:#1]{Corollary~\ref*{clr:#1}}}
\newcommand{\sct}[1]{\hyperref[sec:#1]{Section~\ref*{sec:#1}}}
\newcommand{\subsec}[1]{\hyperref[subsec:#1]{Subsection~\ref*{subsec:#1}}}
\newcommand{\itm}[2]{\hyperref[itm:#1]{#2}}
\newcommand{\clm}[1]{\hyperref[clm:#1]{Claim~\ref*{clm:#1}}}
\newcommand{\rmk}[1]{\hyperref[rmk:#1]{Remark~\ref*{rmk:#1}}}

\definecolor{lightcyan}{RGB}{0.88,1,1}
\definecolor{darkgreen}{RGB}{0, 128, 0}
\definecolor{darkblue}{RGB}{0, 0, 128}

\newcommand{\N}{\mathbb{N}}
\newcommand{\R}{\mathbb{R}}
\newcommand{\C}{\mathbb{C}}

\DeclareMathOperator*{\E}{\mathbb{E}}
\newcommand{\so}{\mathsf{SO}}

\newcommand{\su}{\mathsf{SU}}
\newcommand{\sualg}{\mathfrak{su}}
\renewcommand{\i}{\mathrm{i}}
\newcommand{\e}{\mathrm{e}}

\newcommand{\F}{\mathcal{F}}

\newcommand{\expect}[2]{\E_{\substack{#1}}\!\Br{#2}}
\newcommand{\prob}[2]{\underset{#1}{\mathrm{Pr}}\!\Br{#2}}

\newcommand{\br}[1]{\left(#1\right)}
\newcommand{\Br}[1]{\left[#1\right]}

\newcommand{\tr}[1]{\mathrm{Tr}\!\Br{#1}}
\newcommand{\abs}[1]{\left|#1 \right|}
\newcommand{\norm}[1]{\left\lVert #1 \right\rVert}

\newcommand{\ugroup}[1]{\mathsf{U}\!\br{#1}}

\newcommand{\tv}[1]{\norm{#1}_{\mathrm{TV}}}

\newcommand{\cprim}[1]{\textnormal{\textup{\textsf{#1}}}}

\newcommand{\pru}{\cprim{PRU}}
\newcommand{\prs}{\cprim{PRS}}

\newcommand{\prss}{\cprim{PRSS}}

\usepackage[normalem]{ulem}

\newcommand{\im}{\ensuremath{\mathrm{Im}}}

\newcommand{\btdist}[1]{\mathrm{Beta}(#1)}

\newcommand{\re}{\mathrm{Re}}

\newif\ifnotes
\notestrue

\newcommand{\stframe}[2]{\mathsf{V}_{#1,#2}}

\newcommand{\kackernel}{\mathcal{K}}

\begin{document}
\hypersetup{pageanchor=false}
\pagestyle{empty}

\title{\bf Rapid Mixing of Parallel Kac’s Walk: From Spheres to Stiefel Manifolds}
\date{}
\author[1]{Qian Chen}
\author[2]{Minglong Qin}
\author[3]{Fang Song}
\author[4,5]{Penghui Yao}
\author[4]{Mingnan Zhao}

\affil[1]{Mathematics Research Center, School of Science and Engineering, The Chinese University of Hong Kong, Shenzhen, China}
\affil[ ]{\texttt{chenqian.phys@gmail.com}}
\affil[2]{Centre for Quantum Technologies, National University of Singapore, Singapore}
\affil[ ]{\texttt{mlqin6@gmail.com}}
\affil[3]{Computer Science Department, Portland State University, USA}
\affil[ ]{\texttt{fsong@pdx.edu}}
\affil[4]{State Key Laboratory for Novel Software Technology, Nanjing University, Nanjing 210023, China}
\affil[ ]{\texttt{phyao1985@gmail.com, mingnanzh@gmail.com}}
\affil[5]{Hefei National Laboratory, Hefei 230088, China}

\makeatletter
\patchcmd{\@maketitle}{\vskip 2em}{\vskip -1.5em}{}{\PackageError{main}{Failed to adjust title spacing}{Check the definition of \string\@maketitle.}}
\makeatother

\maketitle
\thispagestyle{empty}
\setcounter{tocdepth}{2}

\allowdisplaybreaks

\vspace{-3em}
\begin{abstract}
Kac's walk is a classical local random walk
whose action on a single real
unit vector in dimension $d$ mixes in total variation in
$\Theta(d\log d)$ sequential steps~\cite{PS17}.
Lu, Qin, Song, Yao, and Zhao introduced
a parallel version of Kac's walk
that mixes a single quantum state in $O(\log d)$ rounds~\cite{LQSY+26}.
After discretizing the randomness and replacing it by suitable
pseudorandom primitives, this parallel walk gives rise to
pseudorandom state scramblers,
and was subsequently shown to yield pseudorandom unitaries~\cite{LQSY+25}.

We study what happens when the parallel Kac's walk acts
simultaneously on $k$ orthonormal quantum states.
We prove that, for any \(1\leq k < d\),
after $O\!\left((k+\log d)\log(d/\varepsilon)\right)$~steps,
the joint distribution of the \(k\) output states is
\(\varepsilon\)-close, in both Wasserstein and total variation distance,
to that obtained by applying a common Haar-random unitary
to the same inputs.
This generalizes the dispersing property of the parallel Kac's walk
from a single quantum state to multiple orthonormal quantum states.
Equivalently, viewing an ordered collection of
\(k\) orthonormal states as a point on the complex Stiefel manifold
$V_{d,k}=\{X\in\mathbb C^{d\times k}:X^\dagger X=I_k\}$,
we show that the parallel Kac’s walk mixes rapidly on \(V_{d,k}\),
with both Wasserstein and total variation mixing times bounded by
$O\!\left((k+\log d)\log(d/\varepsilon)\right)$.
This extends the Wasserstein mixing result of
Pillai, Smith, and Vaikuntanathan for the standard Kac's walk
on real Stiefel manifolds~\cite{PSV26}.
\end{abstract}

\begingroup
\setstretch{0.8}
\setlength{\parskip}{0pt}
\makeatletter
\patchcmd{\l@section}{1.0em \@plus\p@}{0.6em \@plus\p@}{}{\PackageError{main}{Failed to adjust TOC section spacing}{Check the definition of \string\l@section.}}
\makeatother
\tableofcontents
\endgroup

\newpage 

\pagenumbering{arabic}
\pagestyle{plain}
\hypersetup{pageanchor=true}
\begingroup
\setstretch{1.1}

\section{Introduction}
\label{sec:intro}
\addtocontents{toc}{\protect\setcounter{tocdepth}{1}}
Efficiently simulating Haar random states is a fundamental problem
in quantum information theory.
A substantial body of work has studied this problem
from diverse perspectives,
including \emph{quantum state $t$-designs}
\cite{AE07,KG15},
\emph{pseudorandom quantum states} (\prs s) \cite{JLS18,BS19,BS20},
and \emph{stateful quantum simulators} \cite{AMR20}.
Among these, \prs s offer a computational notion of Haar randomness.
They are efficiently preparable quantum states that appear Haar random
to any quantum polynomial-time distinguisher, even when given
polynomially many copies of the same state.
\prs s have found applications
in quantum cryptography \cite{AQY22,MY22},
quantum learning theory \cite{HBC+22},
and quantum gravity \cite{BFV20,YE25}.

Lu, Qin, Song, Yao, and Zhao generalized the notion of \prs s by introducing
\emph{pseudorandom state scramblers} (\prss s) \cite{LQSY+26}.
Whereas \prs s are prepared from a prescribed initial state
(typically the all-zeros state),
\prss s are able to produce pseudorandom states
from arbitrary pure states
using the same family of efficiently implementable unitaries.
Their construction is based on the \emph{parallel Kac's walk},
a parallel version of the well-known Kac's walk on the unit sphere
\cite{Kac56,PS17}.
For an $n$-qubit system, each step of this parallel walk samples a uniformly random perfect
matching on the $d=2^n$ computational basis states and applies an independent
Haar random $\su(2)$ rotation within the span of each matched pair.
They proved that the parallel walk approaches
Haar measure within $\varepsilon$ in \emph{Wasserstein distance}
after $O(n+\log(1/\varepsilon))$ steps.
Together with a suitable discretization
and an efficient implementation
using quantum-secure pseudorandom functions
and pseudorandom permutations,
this yields their construction of \prss s.
In a subsequent work \cite{LQSY+25},
they further showed that the same construction also yields
pseudorandom unitaries \cite{JLS18,MH25}.

Beyond these computational guarantees, the parallel Kac's walk
also exhibits rapid mixing in \emph{total variation distance}
with an $O(n+\log(1/\varepsilon))$  mixing time \cite{LQSY+26}.
This convergence yields a \emph{dispersing property} for their
information-theoretic scrambler:
with sufficient randomness,
the possible outputs from any fixed pure input
form an $\varepsilon$-net of the state space.
Such geometric coverage is not implied by the definitions of
\prs, \prss, or even \pru, which guarantee computational
indistinguishability from Haar randomness.
Moreover, this property seems absent from
all other currently known constructions of \pru s
\cite{MPSY24,MH25,science.adv8590,SMLBH25}.
For this reason, this dispersing property may be useful in applications
where computational indistinguishability alone is insufficient,
such as state and unitary synthesis.

These results make parallel Kac's walk a very useful building block
for generating quantum pseudorandomness.
Its mixing guarantees, however, have so far been
established only for a single input state.
Therefore, it is natural to ask
\begin{center}
    \textit{what happens when the parallel Kac's walk is applied to multiple quantum states simultaneously?}
\end{center}
To make this question more precise,
consider $k$ orthonormal quantum states
$\ket{\psi_1},\dots,\ket{\psi_k}\in\C^d$.
Let $U_T$ be the random unitary obtained after $T$ steps of the parallel
Kac's walk and $V$ be a Haar random unitary on $\ugroup{d}$.
We ask how large $T$ must be for the joint distributions of
\[
    \bigl(U_T\ket{\psi_1},\dots,U_T\ket{\psi_k}\bigr)
    \qquad\text{and}\qquad
    \bigl(V\ket{\psi_1},\dots,V\ket{\psi_k}\bigr)
\]
to be within $\varepsilon$ in total variation distance.
Writing the inputs as columns of a matrix identifies their evolution
with a walk on the \emph{Stiefel manifold}
$\stframe{d}{k}\coloneqq \{X\in\C^{d\times k}:X^\dagger X=I_k\}$,
the space of ordered collections of $k$ orthonormal quantum states.
Its Haar measure $\mu_{d,k}$ is the distribution of $VX$ for any fixed
$X\in\stframe{d}{k}$ and $V\sim\operatorname{Haar}(\ugroup{d})$.

Related mixing problems have been studied extensively
for the original Kac's walk,
in which each step applies a random $\so(2)$ rotation
to a uniformly chosen pair of coordinates.
Pillai and Smith established the optimal $\Theta(d\log d)$ total variation mixing time on unit sphere $\mathbb{S}^{d-1}$ \cite{PS17}.
For the walk on $\mathrm{SO}(d)$, Oliveira proved an
$O(d^2\log d)$ Wasserstein mixing bound \cite{Oliveira09},
and Pillai and Smith recently obtained
a total variation bound of the same order \cite{PS26}.
Between these two settings, Pillai, Smith, and Vaikuntanathan
studied the action of the walk on $k$ orthonormal real vectors
and proved an $O(d(k+\log d)\log(d))$ Wasserstein mixing
bound on the real Stiefel manifold \cite{PSV26},
while leaving open the problem of establishing a comparable bound
in total variation distance.
Such a bound requires further analysis, since Wasserstein convergence
does not in general imply convergence in total variation.

\subsection{Main Results}

In this work,
we show that the parallel Kac's walk jointly scrambles multiple quantum
states, with quantitative bounds in both Wasserstein and total variation
distance.
For an $n$-qubit system, the bound becomes
$O((k+n)(n+\log(1/\varepsilon)))$.
In particular, polynomially many orthonormal inputs can be
jointly scrambled in polynomially many steps of the parallel Kac's walk.

\begin{theorem}[Informal; see \cref{thm:wasserstein_mixing,thm:tv_mixing}]
\label{thm:wasserstein_informal}
\label{thm:tv_informal}
For $\varepsilon\in(0,1)$, the mixing time of the parallel Kac's walk on
$\stframe{d}{k}$ in both Wasserstein and total variation distance is at most
\(
    O\!\left((k+\log d)\log(d/\varepsilon)\right).
\)
\end{theorem}

\vspace{-1.5em}
\paragraph{Proof Overview.}

We prove the mixing time
by coupling two copies $(X_t)_{t\geq 0}$ and $(Y_t)_{t\geq 0}$ of the walk,
with initial states $X_0$ arbitrary
and $Y_0$ drawn from Haar measure $\mu_{d,k}$ on $\stframe{d}{k}$.
As in many coupling proofs including those for Kac's walk
in \cite{PS17,LQSY+26,PS26},
our argument has two stages:
we first bring the two copies close together, and then make them agree
exactly with high probability. The first stage yields our Wasserstein
mixing time bound. In the second stage, we use the coupling lemma
(see \cref{lem:coupling_lemma}) to bound the total variation distance at time $t$
by $\Pr[X_t\ne Y_t]$, the probability that the two copies have not
yet coalesced, and thereby obtain the total variation mixing time bound.

\textbf{\textit{ Step 1: Bringing the two copies close.}}
We extend the coupling technique of \cite{PSV26}
to the parallel walk on the complex Stiefel manifold.
By the local-to-global lemma
(see \cref{lem:local_to_global_wasserstein_contraction}),
it suffices to establish Wasserstein contraction for nearby initial states.
Our main task is therefore to construct an explicit coupling of two chains
started from nearby states $X$ and $Y$
and show that their distance contracts in expectation at each step.
We use a shared uniformly random perfect matching.
For each pair $\ell=(i,j)$ in the matching, let $X_\ell,Y_\ell$
be the $2\times k$ submatrices consisting of rows $i$ and $j$
of $X$ and $Y$, respectively. We choose
\[
    h_\ell\in\operatorname*{arg\,min}_{h\in\mathrm{SU}(2)}
    \|X_\ell-hY_\ell\|_F^2,
\]
so that $h_\ell$ aligns these rows as closely as possible.
We then sample independent Haar random rotations
$R_\ell\in\mathrm{SU}(2)$ and update each pair by
\(
    X_\ell'=R_\ell X_\ell
\)
and
\(
    Y_\ell'=R_\ell h_\ell Y_\ell
\).
\begin{itemize}[topsep=4pt]
    \item For $k\leq d/2$, this coupling ensures that $\|X_t-Y_t\|_F^2$
contracts in expectation by a factor of $1-\Omega(1/(k+\log d))$
at each step with high probability after a short initial period.
Then, we use the local-to-global lemma
to extend this contraction ratio to arbitrary initial states $X$ and $Y$.
Iterating the resulting Wasserstein contraction yields a mixing time of
$O((k+\log d)\log(d/\varepsilon))$ in Wasserstein distance.

    \item For $d/2<k\leq d$, we first consider the full-dimensional case $k=d$
on $\su(d)$, where we prove a one-step Wasserstein contraction
with factor $1-\Omega(1/d)$ and hence a mixing time of
$O(d\log(d/\varepsilon))$. We then transfer this bound to the walk
on $\stframe{d}{k}$ by projecting onto the first $k$ columns.
Since $k=\Theta(d)$ in this range, this yields the desired
$O((k+\log d)\log(d/\varepsilon))$ Wasserstein mixing time.
\end{itemize}

\textbf{\textit{ Step 2: Making the two copies agree.}}
Step 1 brings the two chains to nearby states $X$ and $Y$ with high
probability. Our task is now to couple a further
$T=O((k+\log d)\log d)$ steps so that their final states agree
exactly with high probability.
Let $\mathcal M=(M_1,\dots,M_T)$ be a shared sequence of random
perfect matchings, and let $\mathcal R$ be the sequence of independent
Haar random $\su(2)$ rotations used by the chain started from $X$.
For a fixed $\mathcal M$, write $\Phi_{X,\mathcal M}(\mathcal R)$
for the final state after these $T$ steps.

The key is to show that small changes in the final state
can be realized by controlled perturbations of the rotations.
We establish this through a surjectivity bound
on the Jacobian of $\Phi_{X,\mathcal M}$.
Using this bound, we construct a correction map $\mathcal T$
that adjusts the rotation sequence to compensate for
the difference between $X$ and $Y$.
With high probability over $\mathcal M$ and $\mathcal R$,
applying the corrected sequence $\mathcal T(\mathcal R)$ to $Y$
gives the same final state as applying $\mathcal R$ to $X$:
\[
    \Phi_{Y,\mathcal M}(\mathcal T(\mathcal R))
    =\Phi_{X,\mathcal M}(\mathcal R).
\]
Moreover, when $X$ and $Y$ are sufficiently close, the distribution
of $\mathcal T(\mathcal R)$ is close in total variation to the
product Haar distribution. We therefore use a maximal coupling to
obtain a rotation sequence $\mathcal R_Y$ with exactly the product
Haar distribution that agrees with $\mathcal T(\mathcal R)$ with
high probability. Using $\mathcal R$ for the $X$-chain and
$\mathcal R_Y$ for the $Y$-chain gives a valid coupling in which
the final states agree with high probability.

\subsection{Discussion}
Our result provides a further characterization of the parallel Kac's walk.
Previous work showed that the walk rapidly scrambles an arbitrary fixed pure state and, after suitable discretization and pseudorandom replacement of its randomness, gives rise to pseudorandom state scramblers and pseudorandom unitaries~\cite{LQSY+25,LQSY+26}. Another feature of the ideal walk is its dispersing behavior: for a fixed input state, sufficiently many random choices spread the possible outputs throughout the state space. Our result shows that this phenomenon persists jointly for multiple inputs.
In particular, the joint distribution of outputs approaches
that generated by applying a single Haar random unitary to all of the inputs,
and hence reproduces the correlations imposed by a common random unitary.

Although our main theorem is stated for the complex Stiefel manifold,
our method should also apply to the standard Kac's
walk on the real Stiefel manifold, where each step applies a random
rotation to a single pair of coordinates.
This would yield a
total variation mixing time of
$O(d(k+\log d)\log(d/\varepsilon))$ sequential steps,
matching the Wasserstein mixing bound established in \cite{PSV26}.
Such a result would resolve their open question
of whether the $k$-column walk satisfies the same mixing bound
in total variation distance.

\vspace{-1em}
\paragraph{Organization.}
The rest of the paper is organized as follows.
\Cref{sec:preliminary} introduces the necessary preliminaries.
We define the parallel Kac's walk and state our main mixing result
in \cref{sec:parallel_kac_walk_on_k_columns}.
\Cref{sec:wasserstein_convergence} establishes the Wasserstein mixing bound.
Building on this bound, \cref{sec:tv_convergence} constructs a coupling
of nearby states and proves total variation mixing.
Some technical proofs are deferred to the appendix.

\paragraph{Acknowledgement.}

Q.C. was supported in part by the National Natural Science Foundation of China (Grants Nos. 12271460 and 12341101), Guangdong provincial grants (Grant Nos. 2024A1515011456 and GDZX2403006), and the Shenzhen Fundamental Research Program (Grant No. JCYJ20241202124023031).
M.Q. was supported by the National Research Foundation, Singapore,
through the National Quantum Office, hosted in A*STAR,
under its Centre for Quantum Technologies Funding Initiative (S24Q2d0009).
P.Y. and M.Z. were supported by the National Natural Science Foundation of China (Grant Nos. 62332009 and 12347104), the Quantum Science and Technology-National Science and Technology Major Project (Grant No. 2021ZD0302901), the NSFC/RGC Joint Research Scheme (Grant No. 12461160276), the Natural Science Foundation of Jiangsu Province (No. BK20243060), the Fundamental and Interdisciplinary Disciplines Breakthrough Plan of the Ministry of Education of China (No. JYB2025XDXM118), the ``111 Center'' (No. B26023), and the Fundamental Research Funds for the Central Universities (Grant No. 2026300376).

\paragraph{AI Disclosure.}

OpenAI GPT-6 Astra was used extensively throughout
the preparation of this manuscript,
including for mathematical exploration,
checking proof arguments,
drafting, and language polishing.
In particular, it contributed substantially
to developing the proofs in \cref{sec:tv_convergence}.
The authors verified and refined the proofs,
revised the text, and take full responsibility
for the content of the final manuscript.
\addtocontents{toc}{\protect\setcounter{tocdepth}{2}}

\section{Preliminaries}
\label{sec:preliminary}
\paragraph{Notation.}
Throughout, $\log$ denotes the base-$2$ logarithm.
For a positive integer $n$, let $[n]\coloneqq\{1,\dots,n\}$.
For an event $E$, let $E^{\mathrm c}$ denote its complement
and $\mathbf{1}_E$ its indicator.
We write $I_n$ for the $n\times n$ identity matrix and
$\mathbf e_i$ for the $i$-th standard basis vector.
We write $A\succeq B$ if $A-B$ is positive semidefinite.
For a matrix $A$, let $A^\dagger$ denote its conjugate
transpose
and let $\|A\|_F$ and $\|A\|_{\mathrm{op}}$ denote
its Frobenius norm and operator norm, respectively.
For matrices $A,B$ of the same size, we use the real Frobenius
inner product
$\langle A,B\rangle_F\coloneqq\operatorname{Re}(\tr{A^\dagger B})$.
For a linear map $L$ between matrix spaces, let
$\|L\|_{\mathrm{op}}$ denote its operator norm induced by
the Frobenius norms.

\paragraph{Stiefel Manifolds and Unitary Groups.}

For $1\leq k\leq d$, define the complex Stiefel manifold
$\stframe{d}{k}\coloneqq
\{X\in\C^{d\times k}:X^\dagger X=I_k\}$.
Its elements are matrices with orthonormal columns, and its
tangent space at $X\in\stframe{d}{k}$ is
$\mathsf T_X\stframe{d}{k}
=\{H\in\C^{d\times k}:X^\dagger H+H^\dagger X=0\}$,
which is a real linear space of dimension $2dk-k^2$.
When $k=d$, the Stiefel manifold is the unitary group $\ugroup{d}$,
whose tangent space at $X$ is
$\mathsf T_X\ugroup{d}=\{XA:A\in\C^{d\times d},\ A^\dagger=-A\}$.
The special unitary group is the subgroup
$\su(d)\coloneqq\{U\in\ugroup{d}:\det U=1\}$.
Its Lie algebra is
$\sualg(d)\coloneqq
\{A\in\C^{d\times d}:A^\dagger=-A,\ \tr{A}=0\}$,
which is a real linear space of dimension $d^2-1$,
and its tangent space at $U\in\su(d)$ is
$\mathsf T_U\su(d)=\{UA:A\in\sualg(d)\}$.
We use the following parametrization of $\su(2)$ by the unit
sphere $\mathbb S^3\coloneqq\{u\in\R^4:\|u\|_2=1\}$.
Define
\begin{equation}
    \sigma_0=
    \begin{pmatrix}
        1&0\\
        0&1
    \end{pmatrix}\,,
    \qquad
    \sigma_1=
    \begin{pmatrix}
        0&1\\
        -1&0
    \end{pmatrix}\,,
    \qquad
    \sigma_2=
    \begin{pmatrix}
        0&\i\\
        \i&0
    \end{pmatrix}\,,
    \qquad
    \sigma_3=
    \begin{pmatrix}
        \i&0\\
        0&-\i
    \end{pmatrix} \,.
    \label{eq:basis_for_su2}
\end{equation}
For $u=(u_0,u_1,u_2,u_3)\in\mathbb S^3$, let $h(u)\coloneqq\sum_{z=0}^3u_z\sigma_z$.
Then $h$ is a bijection from $\mathbb S^3$ to $\su(2)$.

We use $D$ to denote the Riemannian distance
on $\stframe{d}{k}$ induced by the Frobenius norm.
The following lemma compares the Frobenius distance with the
Riemannian distance on $\stframe{d}{k}$.
It is the complex analogue of \cite[Lemma~A.4]{PSV26}.
The proof is the same after replacing
transposes with conjugate transposes in the original proof,
so we omit it here.

\begin{lemma}
    \label{lem:frobenius_riemannian_distance_comparison}
    For all $X,Y\in\stframe{d}{k}$,
    \(
        \norm{X-Y}_F
        \leq
        D(X,Y)
    \).
    Moreover, whenever
    $\norm{X-Y}_F\leq 1$,
    \(
        D(X,Y)
        \leq
        \norm{X-Y}_F+2\norm{X-Y}_F^2
    \).
\end{lemma}

\paragraph{Probability Distances and Markov Chain.}
For probability measures $\mu$ and $\nu$ defined on a measurable
space $\br{\Omega,\F}$, their \emph{total variation distance} is
\(
    \norm{\mu-\nu}_{\mathrm{TV}}
    \coloneqq\sup_{A\in\F}
    \abs{\mu\!\br{A}-\nu\!\br{A}}.
\)
A coupling of two probability measures $\mu$ and $\nu$
is a joint distribution with marginals $\mu$ and $\nu$.
For probability measures $\mu$ and $\nu$ on $\stframe{d}{k}$, 
define the Wasserstein distance with respect to $D$ by
\[
    W_D(\mu,\nu)
    \coloneqq
    \inf_{\pi\in\Pi(\mu,\nu)}
    \expect{(X,Y)\sim\pi}{D(X,Y)} \enspace,
\]
where $\Pi(\mu,\nu)$ denotes the set of couplings of $\mu$ and $\nu$.

For a Markov chain on a state space $\Omega$ with transition
kernel $K$, let $K^t(x,\cdot)$ denote its distribution after
$t$ steps when started from $x\in\Omega$.
A probability measure $\pi$ is stationary if the chain
has distribution $\pi$ at every time when initialized
with distribution $\pi$.
For $\varepsilon\in(0,1)$, its total variation mixing time is
$t_{\mathrm{mix}}(\varepsilon)
\coloneqq\min\{t\in\mathbb Z_{\geq0}:
\sup_{x\in\Omega}\|K^t(x,\cdot)-\pi\|_{\mathrm{TV}}
\leq\varepsilon\}$.

We will use the following lemma
that reduces global Wasserstein contraction to a
contraction estimate for sufficiently close initial states.

\begin{lemma}[Local-to-global Wasserstein contraction {\cite[Lemma~3.3]{PSV26}}]
    \label{lem:local_to_global_wasserstein_contraction}
    Let $(E,D)$ be a compact geodesic metric space and let $K$ be a
    Markov kernel on $E$. Let $W_D$ denote the Wasserstein distance
    with cost $D$. Suppose that, for some $\rho>0$ and
    $\lambda\in[0,1)$,
    \[
        W_D\bigl(K(x,\cdot),K(y,\cdot)\bigr)
        \leq
        \lambda \cdot D(x,y) \enspace,
    \]
    whenever $D(x,y)\leq\rho$. Then, for all $x,y\in E$,
    \[
        W_D\bigl(K(x,\cdot),K(y,\cdot)\bigr)
        \leq
        \lambda \cdot D(x,y) \enspace.
    \]
\end{lemma}

We will also use the following coupling bound for total variation distance.
\begin{lemma}[Coupling lemma, {\cite[Theorem 5.4]{LPW09}}]
\label{lem:coupling_lemma}
Let $K$ be a Markov transition kernel on a state space $\Omega$,
with stationary probability measure $\pi$.
Let $(X_t)_{t\geq0}$ and $(Y_t)_{t\geq0}$ be Markov chains with
transition kernel $K$, started at $X_0=x\in\Omega$ and $Y_0\sim\pi$,
respectively. Then, for any coupling of the two chains and every
integer $t\geq0$,
\(
    \tv{K^t(x,\cdot)-\pi}
    \leq\prob{}{X_t\ne Y_t}.
\)
\end{lemma}

\section{The Parallel Kac's Walk}
\label{sec:parallel_kac_walk_on_k_columns}
We study the mixing time of the parallel Kac's walk \cite{LQSY+26} on $\stframe{d}{k}$,
which is a Markov chain defined as follows.

\begin{definition}[Parallel Kac's walk on $\stframe{d}{k}$]
    \label{def:parallel_kac_walk_on_k_columns}
    Let $d$ be even and set $m=d/2$.
    The parallel Kac's walk on $\stframe{d}{k}$ is a Markov chain
    $(X_t\in \stframe{d}{k})_{t \ge 0}$
    which evolves according to the following transition rule at each time step $t$:
    \begin{enumerate}
        \item Sample a random perfect matching of $[d]$ uniformly at random, denoted by
        \[
            M = \{(i_1,j_1),\dots,(i_m,j_m)\}\enspace.
        \]
        \item Sample $m$ independent Haar random $\su(2)$ matrices, denoted by $R_1,\dots,R_m$. For each $\ell\in[m]$, view $R_\ell$ as an element of $SU(d)$
        that acts on the two-dimensional plane spanned by $\mathbf e_{i_\ell}$ and $\mathbf e_{j_\ell}$
        and leaves all other coordinates unchanged.
        \item 
        The next state is then given by
        \(
            X_{t+1}
            =
            \br{\prod_{\ell=1}^{m}R_\ell}\cdot X_t
        \).
    \end{enumerate}
\end{definition}

Let $\mu_{d,k}$ be the Haar measure on $\stframe{d}{k}$ and
$\kackernel$ the transition kernel of the parallel Kac's walk.
Define its total variation mixing time by
\[
    t_{\mathrm{mix}}(\varepsilon)
    =\min\left\{t\geq0:
        \sup_{X\in\stframe{d}{k}}\tv{\kackernel^t(X,\cdot)-\mu_{d,k}}
        \leq\varepsilon\right\}\enspace.
\]
Our main result is the following bound on the mixing time.

\begin{restatable}{theorem}{tvMixingTheorem}
    \label{thm:tv_mixing}
    For every even $d\geq2$, $1\leq k<d$, and
    $\varepsilon\in(0,1)$, the parallel Kac's walk on $\stframe{d}{k}$ satisfies
    \[
        t_{\mathrm{mix}}(\varepsilon)
        ~\leq~2^{20}(k+\log d)\log(d/\varepsilon)\enspace.
    \]
\end{restatable}

Note that every matrix in $\stframe{d}{d-1}$
has a unique completion to a unitary in $\su(d)$.
This correspondence immediately gives the following corollary.

\begin{corollary}
For every even $d\geq2$ and $\varepsilon\in(0,1)$,
the total variation mixing time of the parallel Kac's walk
on $\su(d)$ is at most
$O(d\log(d/\varepsilon))$.
\end{corollary}

We will prove \cref{thm:tv_mixing} in \cref{sec:tv_convergence}
using a coupling argument.
Before doing so, we establish a Wasserstein mixing bound in
\cref{sec:wasserstein_convergence}, which allows us to bring
two walks close before coupling their final states to agree.

\section{Wasserstein Mixing}
\label{sec:wasserstein_convergence}
In this section, we prove that the parallel Kac's walk mixes
in Wasserstein distance after $O((k+\log d)\log(d/\varepsilon))$ steps.
The following theorem gives a quantitative bound.

\begin{restatable}[Wasserstein mixing]{theorem}{wassersteinMixingTheorem}
\label{thm:wasserstein_mixing}
Let $d\geq2$ be even, $1\leq k\leq d$, and let
$\kackernel$ be the transition kernel of the parallel Kac's walk
on $\stframe{d}{k}$.
For all $X,Y\in\stframe{d}{k}$ with $\det X=\det Y$ when $k=d$,
every $\varepsilon\in(0,1)$, and every integer
\(
    t\geq4000(k+\log d)\log(d/\varepsilon),
\)
we have
\(
    W_D\!\left(
        \kackernel^t(X,\cdot),
        \kackernel^t(Y,\cdot)
    \right)
    \leq\varepsilon.
\)
\end{restatable}

We prove the theorem by treating two ranges of $k$ separately.
In \cref{subsec:wasserstein_low_rank}, we consider $k\leq d/2$
and combine a one-step Frobenius contraction with bounds on
the maximum row norm to obtain Wasserstein contraction over
blocks of steps.
In \cref{subsec:wasserstein_high_rank}, we consider $d/2<k\leq d$,
first establish contraction on $\su(d)$, and then transfer
the resulting mixing time to $\stframe{d}{k}$.
Finally, we use the bounds established in these two cases to prove
\cref{thm:wasserstein_mixing} in \cref{subsec:wasserstein_mixing_bound}.

\subsection{\texorpdfstring{The Case $k\leq d/2$}{The Case k ≤ d/2}}
\label{subsec:wasserstein_low_rank}

For $k\leq d/2$, our main result is a $1/d$ contraction after $O((k+\log d)\log d)$ steps described in the following lemma.
We prove this lemma at the end of this subsection.

\begin{restatable}{lemma}{lowrankcontraction}
\label{lem:low_rank_contraction}
Let $d\geq 2$ be even, $1\leq k\leq d/2$,
and $\kackernel$ be the transition kernel of the parallel Kac's walk on $\stframe{d}{k}$.
Let
\(
    T_0
    =
    \left\lceil 1000(k+\log d)\log d\right\rceil
\).
Then, for every
$X,Y\in\stframe{d}{k}$,
\[
    W_D\!\left(\kackernel^{T_0}(X,\cdot),\kackernel^{T_0}(Y,\cdot)\right)
    ~\leq~
    \frac{1}{d}\cdot D(X,Y) \enspace.
\]
\end{restatable}

By iterating \cref{lem:low_rank_contraction}, we immediately obtain the following theorem.

\begin{theorem}
    \label{thm:low_rank_wasserstein_mixing}
    Let $d\geq 2$ be even, $1\leq k\leq d/2$,
    and $\kackernel$ be the transition kernel of
    the parallel Kac's walk on $\stframe{d}{k}$.
    For every $t\geq 0$
    and $X,Y \in \stframe{d}{k}$,
    we have
    \[
        W_D\!\left(\kackernel^{t}(X,\cdot),\kackernel^{t}(Y,\cdot)\right)
        ~\leq~
        d^{-\left\lfloor\tfrac{t}{\left\lceil 1000(k+\log d)\log d\right\rceil}\right\rfloor}\cdot D(X,Y) \enspace.
    \]
\end{theorem}

\subsubsection{One-Step Frobenius Contraction}
To prove \cref{lem:low_rank_contraction},
we begin with a coupling for nearby initial states that never
increases the Frobenius distance.
We show that the Frobenius distance contracts in expectation
at a rate related to the squared row norms of the initial states.

\begin{lemma}
    \label{lem:frobenius_contraction}
    Suppose that $k\leq d/2$.
    For every $X,Y \in \stframe{d}{k}$ such
    that $\norm{X-Y}_{F}\leq 1$ and every row of $X$ and $Y$ has
    squared $2$-norm at most $\Delta$,
    there exists a coupling of one step of parallel Kac's walk from
    $(X,Y)$ to $(X',Y')$ such that $\norm{X'-Y'}_{F} \leq \norm{X-Y}_{F}$ and
    \[
	    	\expect{}{ \norm{X'-Y'}_{F}^{2} } \leq \br{1 - \frac{3}{16\Delta(d-1)}}\cdot \norm{X-Y}_{F}^{2}\enspace.
    \]
\end{lemma}

\begin{proof}
We construct the coupling as follows:
\begin{itemize}
    \item Sample a uniformly random perfect matching $M=\{(i_1,j_1),\dots,(i_m,j_m)\}$ of $[d]$.
    \item For each $\ell\in[m]$,
    let \(X_\ell\) and \(Y_\ell\) denote the \(2\times k\) submatrices
    obtained by restricting \(X\) and \(Y\), respectively,
    to rows \(i_\ell\) and \(j_\ell\).
    Recall the four matrices \(\sigma_0,\dots,\sigma_3\)
    defined in \cref{eq:basis_for_su2},
    and set $p_\ell\coloneqq(p_{\ell,0},\dots,p_{\ell,3})$
    where $p_{\ell,z} = \re\!\br{\tr{X_\ell^\dagger\sigma_zY_\ell}}$
    for $z\in\{0,1,2,3\}$. 
    If $p_\ell\neq0$,
    set $h_\ell = \sum_{z=0}^{3} \frac{p_{\ell,z}}{\norm{p_\ell}_2}\sigma_z$;
    if $p_\ell=0$, set $h_\ell=\sigma_0$.
    \item Sample $m$ independent Haar random $\su(2)$ matrices
    $R_1,\dots,R_m$. For each $\ell\in[m]$, view $R_\ell$ and $h_\ell$
    as elements of $SU(d)$ that act on the two-dimensional plane spanned
    by $\mathbf e_{i_\ell}$ and $\mathbf e_{j_\ell}$ and fix all other coordinates.
    \item Set
    \(
        X'
        =
        \br{\prod_{\ell=1}^{m}R_\ell}\cdot X
    \)
    and
    \(
        Y'
        =
        \br{\prod_{\ell=1}^{m}(R_\ell h_\ell)}\cdot Y
    \).
\end{itemize}
This is a valid coupling since 
$R_\ell h_\ell$ is Haar random by the right invariance of Haar measure.

We first show that this coupling does not increase the Frobenius distance. 
Note that
\begin{align*}
    \norm{X-Y}_F^2-\norm{X'-Y'}_F^2
    ~&=~
    \sum_{\ell = 1}^{m}
    \br{\norm{X_\ell-Y_\ell}_F^2-\norm{R_\ell X_\ell-R_\ell h_\ell Y_\ell}_F^2} \\
    ~&=~
    \sum_{\ell = 1}^{m}
    \br{\norm{X_\ell-Y_\ell}_F^2-\norm{X_\ell-h_\ell Y_\ell}_F^2} \enspace.
\end{align*}
Using the fact that
$\norm{A - B}_F^2 = \norm{A}_F^2 + \norm{B}_F^2 -2\cdot\re\!\br{\tr{A^\dagger B}}$,
we have
\begin{align}
    \norm{X-Y}_F^2-\norm{X'-Y'}_F^2
    ~&=~
    2\cdot\sum_{\ell = 1}^{m}
    \br{\textstyle\re\!\br{\tr{X_\ell^\dagger h_\ell Y_\ell}} - \re\!\br{\tr{X_\ell^\dagger Y_\ell}}} \nonumber\\
    ~&=~
    2\cdot\sum_{\ell = 1}^{m}
    \br{\sum_{z=0}^{3} \frac{p_{\ell,z}}{\norm{p_\ell}_2} \cdot
    {\textstyle\re\!\br{\tr{X_\ell^\dagger \sigma_z Y_\ell}} - \re\!\br{\tr{X_\ell^\dagger Y_\ell}}}} \nonumber\\
    ~&=~
    2\cdot\sum_{\ell = 1}^{m}
    \br{\norm{p_\ell}_2 - p_{\ell,0}} \enspace. \label{eq:contraction_of_frobenius_distance}
\end{align}
The second equality assumes $p_\ell\neq0$ for every $\ell$.
If $p_\ell=0$, then $h_\ell=I_2$ and
$\re\!\br{\tr{X_\ell^\dagger h_\ell Y_\ell}}
=p_{\ell,0}=\|p_\ell\|_2=0$,
so the final equality remains valid.
Since $\norm{p_\ell}_2 - p_{\ell,0} \geq 0$,
$\norm{X'-Y'}_F \leq \norm{X-Y}_F$ holds as desired.

We next show that the expected decrease in the squared Frobenius distance
is bounded below by a positive fraction of the original distance.
Note that if $\norm{p_\ell}_2 + p_{\ell,0} > 0$, we have
\begin{align}\label{eq:lower_bound_on_norm_p_minus_p0}
    \norm{p_\ell}_2-p_{\ell,0}
    ~=~
    \frac{\norm{p_\ell}_2^2-p_{\ell,0}^2}
         {\norm{p_\ell}_2+p_{\ell,0}} 
    ~=~
    \frac{ \sum_{z=1}^{3}p_{\ell,z}^2}
         {\norm{p_\ell}_2+p_{\ell,0}} 
    ~\geq~
    \frac{1}{8\Delta}\cdot
    \sum_{z=1}^{3}p_{\ell,z}^2\enspace,
\end{align}
where the last inequality uses the fact that
$\norm{p_\ell}_2+p_{\ell,0}\leq 2\norm{p_\ell}_2$
and that $\abs{p_{\ell,z}}\leq \norm{X_\ell}_F \cdot \norm{Y_\ell}_F \leq 2\Delta$ for all $z\in\{1,2,3\}$ by the Cauchy--Schwarz inequality.
If $\norm{p_\ell}_2+p_{\ell,0}=0$,
then $p_{\ell,z}=0$ for $z\in\{1,2,3\}$,
so \cref{eq:lower_bound_on_norm_p_minus_p0} holds trivially.
It remains to lower bound the sum on the right-hand side over all pairs.
We use the following claim which is proved in \cref{app:pair_observation}.
\begin{restatable}{claim}{pairObservationClaim}
    \label{clm:pair_observation}
    Suppose that $k\leq d/2$ and $X,Y\in\stframe{d}{k}$ satisfy
    $\norm{X-Y}_F\leq1$. For every $1\leq i<j\leq d$, let $X_{(i,j)}$
    and $Y_{(i,j)}$ denote the $2\times k$ submatrices obtained by
    restricting $X$ and $Y$ to rows $i$ and $j$, and define
    \(
        p_{(i,j),z}
        =
        \re\!\br{\tr{X_{(i,j)}^\dagger\sigma_zY_{(i,j)}}}
    \)
    for $z\in\{0,1,2,3\}$. Then, we have
    \[
        \sum_{i<j}\sum_{z=1}^{3}p_{(i,j),z}^2
        ~\geq~
        \frac34\cdot\norm{X-Y}_F^2\enspace.
    \]
\end{restatable}

Combining \cref{eq:contraction_of_frobenius_distance,eq:lower_bound_on_norm_p_minus_p0}, we have
\begin{align*}
    \norm{X-Y}_F^2-\norm{X'-Y'}_F^2
    ~&\geq~
    \frac{1}{4\Delta}\cdot
    \sum_{\ell = 1}^{m}\sum_{z=1}^{3}p_{\ell,z}^2 \enspace.
\end{align*}
By \cref{clm:pair_observation} and the fact that every fixed pair
$(i,j)$ belongs to a uniformly random perfect matching with
probability $1/(d-1)$, we have
\begin{align}
    \expect{}{
        \norm{X-Y}_F^2-
        \norm{X'-Y'}_F^2
    }
    ~&\geq~
    \frac{1}{4\Delta}\cdot
    \expect{}{
        \sum_{(i,j)\in M}\sum_{z=1}^{3}p_{(i,j),z}^2
    } \nonumber \\
    ~&=~
    \frac{1}{4\Delta(d-1)}\cdot
    \sum_{i<j}\sum_{z=1}^{3}p_{(i,j),z}^2 \nonumber\\
    ~&\geq~
    \frac{3}{16\Delta(d-1)}\cdot
    \norm{X-Y}_F^2\enspace. \label{eq:expected_contraction}
\end{align}
Equivalently,
\begin{align*}
    \expect{}{\norm{X'-Y'}_F^2}
    ~&\leq~
    \br{1-\frac{3}{16\Delta(d-1)}}\cdot
    \norm{X-Y}_F^2\enspace. \qedhere
\end{align*}
\end{proof}

The estimate above becomes stronger when the
maximum row norm is small.
In particular, a squared row-norm bound of order
$(k+\log d)/d$ gives a one-step contraction rate of order
$1/(k+\log d)$.
The following lemma provides the row-norm control needed
to obtain this rate.

\begin{lemma}
    \label{lem:maximum_row_norm}
    Let $(X_t)_{t\geq0}$ be the parallel Kac's walk on $\stframe{d}{k}$
    started from an arbitrary initial state $X_0\in\stframe{d}{k}$.
    For every $\delta\in(0,1)$ and every integer
    \(
        t\geq 30\log\br{2d/\delta}
    \),
    we have
    \[
        \prob{}{
            \max_{i\in[d]}\norm{\mathbf e_i^\dagger X_t}_2^2
            \geq
            4\cdot\frac{k+\log\br{d/\delta}}{d}
        }
        \leq
        \delta\enspace.
    \]
\end{lemma}

\begin{proof}
For each $s\in\{0,\dots,t-1\}$, let $G_s\in SU(d)$ be the random update
matrix at time $s$, so that $X_{s+1}=G_sX_s$ and $X_t=G_{t-1}\cdots G_0X_0$.
Since $X_0^\dagger X_0=I_k$, we have that for any fixed $i\in[d]$,
\begin{align}
    \norm{\mathbf e_i^\dagger X_t}_2^2
    =
    \norm{X_0^\dagger G_0^\dagger\cdots G_{t-1}^\dagger \mathbf e_i}_2^2
    =
    \norm{X_0X_0^\dagger G_0^\dagger\cdots G_{t-1}^\dagger \mathbf e_i}_2^2
    \enspace. \label{eq:parallel_kac_walk_row_norm}
\end{align}
Note that the vector $Z_t\coloneqq G_0^\dagger\cdots G_{t-1}^\dagger \mathbf e_i$
has the same distribution as the $t$-step parallel Kac's walk on the complex unit sphere
\cite[Section 4.2]{LQSY+26}
started from the unit vector $\mathbf e_i$.
Applying \cite[Theorem 4.8]{LQSY+26}\footnote{Although Theorem 4.8 is stated for sufficiently large dimensions,
its proof does not use this assumption, so the result holds for every even
dimension $d\geq2$.}
with $c=t/(10\log d)-1$,
there exists a coupling of $Z_t$ with a Haar random vector $U$ such that
\[
    \expect{}{\norm{Z_t-U}_2}
    ~\leq~
    2^{-t/10+\log d}
    ~\leq~
    \frac{\delta^3}{8d^2}
    \enspace.
\]
Using Markov's inequality yields
\begin{align}
    \prob{}{\norm{Z_t-U}_2\geq
        \sqrt{\frac{k+\log\br{d/\delta}}{16d}}}
    ~&\leq~
    \sqrt{\frac{16d}{k+\log\br{d/\delta}}}
    \cdot
    \frac{\delta^3 }{8d^2}
    ~\leq~
    \frac{\delta}{2d}\enspace.
    \label{eq:parallel_vector_coupling_error}
\end{align}

We now control the projection of the Haar random vector $U$ onto the
column space of $X_0$.
Since $(X_0X_0^\dagger)^2 = X_0X_0^\dagger$ and $ (X_0X_0^\dagger)^\dagger = X_0X_0^\dagger $,
we have that $X_0X_0^\dagger$ is a rank-$k$ orthogonal projection.
Hence, by \cite[Appendix E.4]{RSS22},
we have that $\norm{X_0X_0^\dagger U}_2^2$ follows the $\btdist{k,d-k}$ distribution.
Let $q=\lceil\log(2d/\delta)\rceil$ and the moment formula for the beta distribution gives
\[
    \expect{}{\norm{X_0X_0^\dagger U}_2^{2q}}
    ~=~
    \prod_{j=0}^{q-1}\frac{k+j}{d+j}
    ~\leq~
    \br{\frac{k+q-1}{d}}^q\enspace.
\]
Again, using Markov's inequality therefore gives
\begin{align}
    \prob{}{\norm{X_0X_0^\dagger U}_2^2\geq
        \frac{3\br{k+\log\br{d/\delta}}}{d}}
    ~\leq~
    \br{\frac{k+q-1}{3\br{k+\log\br{d/\delta}}}}^q 
    ~\leq~
    2^{-q} 
    ~\leq~
    \frac{\delta}{2d}\enspace,
    \label{eq:haar_projection_tail}
\end{align}
where the second inequality holds
because $ k+q-1 \leq k + \log(d/\delta) + 1 \leq \frac{3}{2}(k+\log(d/\delta)) $.

Suppose that neither of the two events in
\cref{eq:parallel_vector_coupling_error,eq:haar_projection_tail}
occurs.
By \cref{eq:parallel_kac_walk_row_norm} and the triangle inequality, we have
\begin{align*}
    \norm{\mathbf e_i^\dagger X_t}_2^2
    ~=~
    \norm{X_0X_0^\dagger Z_t}_2^2
    ~\leq~
    \br{\norm{X_0X_0^\dagger U}_2+\norm{Z_t-U}_2}^2 
    &<~
    \frac{4\br{k+\log\br{d/\delta}}}{d}\enspace.
\end{align*}
Therefore, for a fixed row $i\in[d]$, we have
\[
	\prob{}{
        \norm{\mathbf e_i^\dagger X_t}_2^2 \geq \frac{4\br{k+\log\br{d/\delta}}}{d}
    }
    ~\leq~
    \frac{\delta}{d}\enspace.
\]
The lemma then follows by taking a union bound over all $i\in[d]$.
\end{proof}

\subsubsection{Multi-Step Wasserstein Contraction}

We now combine the one-step contraction and row-norm estimates
to prove \cref{lem:low_rank_contraction},
which is restated below for convenience.

\lowrankcontraction*

\begin{proof}
Apply \cref{lem:local_to_global_wasserstein_contraction}
to the $T_0$-step transition kernel.
It suffices to prove the desired contraction
when $D(X,Y)\leq1$.
Fix such $X,Y\in\stframe{d}{k}$,
and we use the coupling in \cref{lem:frobenius_contraction}
for $T_0$ steps of the parallel Kac's walk.
Let
\(
    t_0
    =
    \left\lceil30\log(2\cdot d^7)\right\rceil.
\)
For every $t\in\{t_0,\dots,T_0-1\}$, define
\[
    \mathcal{G}_t
    =
    \left\{
        \max_{i\in[d]}\norm{\mathbf e_i^\dagger X_t}_2^2
        ~\leq~ 4\cdot \frac{k+7\log d}{d}
        \quad\wedge\quad
        \max_{i\in[d]}\norm{\mathbf e_i^\dagger Y_t}_2^2
        ~\leq~ 4\cdot \frac{k+7\log d}{d}
    \right\}\enspace.
\]
Applying \cref{lem:maximum_row_norm} with $\delta=d^{-6}$
and taking a union bound gives
\(
    \prob{}{\mathcal{G}_t}
    \geq
    1 - 2\cdot d^{-6}
\).

Since the coupling does not increase the Frobenius distance,
using \cref{lem:frobenius_riemannian_distance_comparison},
we have
\(
    \norm{X_t-Y_t}_F \leq \norm{X-Y}_F \leq D(X,Y) \leq 1
\)
for all $t\in\{0,\dots,T_0\}$.
Therefore, 
for $t \in \{t_0,\dots,T_0-1\}$,
on the event $\mathcal{G}_t$,
\cref{lem:frobenius_contraction} gives
\begin{align*}
    \expect{}{\left.\norm{X_{t+1}-Y_{t+1}}_F^2\,\right\vert X_t,Y_t}
    ~&\leq
    \left(1-\frac{3d}{64(k+7\log d)(d-1)}\right) \cdot \norm{X_{t}-Y_{t}}_F^2 \\
    ~&\leq
     \left(1-\frac{1}{160(k+\log d)}\right)\cdot \norm{X_{t}-Y_{t}}_F^2 \enspace.
\end{align*}
On the complement of $\mathcal{G}_t$,
the coupling does not increase the Frobenius distance.
Therefore,
\begin{align*}
    &\expect{}{\norm{X_{t+1}-Y_{t+1}}_F^2} \\
    &\qquad\leq~
    \left(1-\frac{1}{160(k+\log d)}\right)\cdot
    \expect{}{
        \mathbf{1}_{\mathcal{G}_t} \cdot
        \norm{X_t-Y_t}_F^2
    }  
    ~+~
    \expect{}{
        \mathbf{1}_{\mathcal{G}_t^{\mathrm c}} \cdot
        \norm{X_t-Y_t}_F^2
    } \\
    &\qquad=~
    \left(1-\frac{1}{160(k+\log d)}\right)\cdot
    \expect{}{\norm{X_t-Y_t}_F^2}
    ~+~
    \frac{1}{160(k+\log d)}\cdot
    \expect{}{
        \mathbf{1}_{\mathcal{G}_t^{\mathrm c}} \cdot
        \norm{X_t-Y_t}_F^2
    } \\
    &\qquad\leq~
    \left(1-\frac{1}{160(k+\log d)}\right)\cdot
    \expect{}{\norm{X_t-Y_t}_F^2}
    +
    \frac{2\cdot d^{-6}}{160(k+\log d)}\cdot
    \norm{X-Y}_F^2
    \enspace.
\end{align*}
Iterating this bound yields
\begin{align*}
    \expect{}{\norm{X_{T_0}-Y_{T_0}}_F^2}
    ~&\leq~
    \left(1-\frac{1}{160(k+\log d)}\right)^{T_0-t_0} \cdot~
    \expect{}{\norm{X_{t_0}-Y_{t_0}}_F^2} \\
    &\qquad+~
    \frac{2\cdot d^{-6}}{160(k+\log d)}\cdot
    \norm{X-Y}_F^2 ~\cdot
    \sum_{j=0}^{T_0-t_0-1}
    \left(1-\frac{1}{160(k+\log d)}\right)^j \\
    ~&\leq~
    \left[
        \left(1-\frac{1}{160(k+\log d)}\right)^{T_0-t_0}
        +~2 \cdot d^{-6}
    \right]\norm{X-Y}_F^2 \\
    ~&\leq~
    \left[
        \e^{-\frac{T_0-t_0}{160(k+\log d)}}
        +2 \cdot d^{-6}
    \right]\norm{X-Y}_F^2
    \enspace.
\end{align*}
Since
\(
    T_0-t_0
    \geq
    800(k+\log d)\log d,
\)
we have
\(
    \frac{T_0-t_0}{160(k+\log d)}
    \geq
    5\log d
    \geq
    5\ln d
\).
Therefore,
\begin{align}
    \expect{}{\norm{X_{T_0}-Y_{T_0}}_F^2}
    ~&\leq~
    \left(d^{-5}+2 \cdot d^{-6}\right)\norm{X-Y}_F^2
    ~\leq~
    \frac{1}{4\cdot d^2}\norm{X-Y}_F^2 \enspace.
    \label{eq:contraction_of_frobenius_distance_after_T0_steps}
\end{align}
Applying \cref{lem:frobenius_riemannian_distance_comparison}
and Jensen's inequality, we then obtain
\begin{align*}
    \expect{}{D(X_{T_0},Y_{T_0})}
    ~&\leq~
    \sqrt{
        \expect{}{\norm{X_{T_0}-Y_{T_0}}_F^2}
    }
    +
    2\expect{}{\norm{X_{T_0}-Y_{T_0}}_F^2} \\
    ~&\leq~
    \frac{1}{2d} \cdot \norm{X-Y}_F
    +
    \frac{1}{2\cdot d^2} \cdot \norm{X-Y}_F^2 
    ~\leq~
    \frac{1}{d} \cdot D(X,Y)
    \enspace,
\end{align*}
where
the second inequality follows from \cref{eq:contraction_of_frobenius_distance_after_T0_steps},
and the last inequality uses
\cref{lem:frobenius_riemannian_distance_comparison}
and the assumption that $D(X,Y)\leq 1$.
Therefore, we have
\[
    W_D\!\left(
        \kackernel^{T_0}(X,\cdot),
        \kackernel^{T_0}(Y,\cdot)
    \right)
    ~\leq~
    \frac{1}{d}\cdot D(X,Y)
    \enspace.
\]
Applying
\cref{lem:local_to_global_wasserstein_contraction}
completes the proof.
\end{proof}

\subsection{\texorpdfstring{The Case $d/2<k\leq d$}{The Case d/2 < k ≤ d}}
\label{subsec:wasserstein_high_rank}
For $d/2<k\leq d$, this subsection establishes a Wasserstein convergence bound
of $O(d\log(d/\varepsilon))$ steps.
Since $d<2k$, this suffices for the dependence on $k$
in \cref{thm:wasserstein_mixing}.
We first prove one-step contraction on $\su(d)$ and then
transfer the result to $\stframe{d}{k}$ by projecting onto the first $k$ columns.

\subsubsection{\texorpdfstring{Wasserstein Contraction on $\su(d)$}{Wasserstein Contraction on SU(d)}}

We establish one-step Wasserstein contraction on $\su(d)$,
using a coupling similar to that in \cref{lem:frobenius_contraction}.
Since that lemma requires $k\leq d/2$, we give a separate
contraction proof for $\su(d)$.

\begin{lemma}
    \label{lem:su_parallel_contraction}
    Let $d\geq2$ be even and $\kackernel$ be the transition kernel of the parallel Kac's walk on $\su(d)$.
    For every $X,Y\in\su(d)$,
    \[
        W_D\!\left(
            \kackernel(X,\cdot),
            \kackernel(Y,\cdot)
        \right)
        ~\leq~
        \left(1-\frac{1}{4(d-1)}\right)\cdot
        D(X,Y)\enspace.
    \]
\end{lemma}

\begin{proof}
By \cref{lem:local_to_global_wasserstein_contraction},
it suffices to consider $X,Y\in\su(d)$ satisfying
$D(X,Y)\leq\frac{1}{16(d-1)}$
and prove the desired contraction.
To this end, we use the same coupling in \cref{lem:frobenius_contraction}.
Given $X$ and $Y$, we sample a uniformly random perfect matching
and define $X_\ell$, $Y_\ell$, $p_\ell$, and $h_\ell$
exactly as in the proof of \cref{lem:frobenius_contraction}.
We then sample independent Haar random matrices in $\su(2)$,
denoted by $R_1,\dots,R_m$, and set
\(
    X'
    =
    \br{\prod_{\ell=1}^{m}R_\ell}\cdot X
\)
and
\(
    Y'
    =
    \br{\prod_{\ell=1}^{m}(R_\ell h_\ell)}\cdot Y
\).
As in the proof of \cref{lem:frobenius_contraction}, this is a valid
coupling because $R_\ell h_\ell$ is Haar random by the right
invariance of Haar measure.
We show that the coupling satisfies
\[
	\expect{}{ D(X',Y') }
    ~\leq~
    \left(1-\frac{1}{4(d-1)}\right)\cdot D(X,Y)\enspace.
\]

By \cref{eq:contraction_of_frobenius_distance}, we have
\begin{equation}
    \norm{X-Y}_F^2-\norm{X'-Y'}_F^2
    ~=~
    2\sum_{\ell=1}^m
    \left(\norm{p_\ell}_2-p_{\ell,0}\right) \enspace.
    \label{eq:su_exact_frobenius_decrease}
\end{equation}
Since $\norm{p_\ell}_2-p_{\ell,0}\geq0$, we have
$\norm{X'-Y'}_F\leq\norm{X-Y}_F$.
Since $X$ and $Y$ are unitaries,
$\norm{X_\ell}_F=\norm{Y_\ell}_F=\sqrt2$.
Thus, by the Cauchy--Schwarz inequality,
$ p_{\ell,0} \leq \norm{p_\ell}_2 = \re\!\br{\tr{X_\ell^\dagger h_\ell Y_\ell}} \leq \norm{X_\ell}_F\cdot\norm{h_\ell Y_\ell}_F \leq 2 $.
If $\norm{p_\ell}_2+p_{\ell,0}>0$, then, as in
\cref{eq:lower_bound_on_norm_p_minus_p0},
\begin{equation}
    \norm{p_\ell}_2-p_{\ell,0}
    ~=~
    \frac{\sum_{z=1}^3p_{\ell,z}^2}
         {\norm{p_\ell}_2+p_{\ell,0}}
    ~\geq~
    \frac14\sum_{z=1}^3p_{\ell,z}^2.
    \label{eq:su_p_decrease_lower_bound}
\end{equation}
If $\norm{p_\ell}_2+p_{\ell,0}=0$,
then $p_{\ell,z}=0$ for $z\in\{1,2,3\}$,
so \cref{eq:su_p_decrease_lower_bound} holds trivially.
It remains to lower bound the sum on the right-hand side over all pairs.
We use the following claim, which is proved in
\cref{app:su_pair_observation}.
\begin{restatable}{claim}{suPairObservationClaim}
    \label{clm:su_pair_observation}
    Let $d\geq2$ and suppose that $X,Y\in\su(d)$ satisfy
    $D(X,Y)\leq \frac{1}{16(d-1)}$. For every $1\leq i<j\leq d$, let
    $X_{(i,j)}$ and $Y_{(i,j)}$ denote the $2\times d$ submatrices
    obtained by restricting $X$ and $Y$ to rows $i$ and $j$, and define
    \(
        p_{(i,j),z}
        =
        \re\!\br{\tr{X_{(i,j)}^\dagger\sigma_zY_{(i,j)}}}
    \)
    for $z\in\{0,1,2,3\}$. Then, we have
    \[
        \sum_{i<j}\sum_{z=1}^3p_{(i,j),z}^2
        ~\geq~
        \frac32\cdot\norm{X-Y}_F^2\enspace.
    \]
\end{restatable}

Combining
\cref{eq:su_exact_frobenius_decrease,eq:su_p_decrease_lower_bound},
we have
\begin{align*}
    \norm{X-Y}_F^2-\norm{X'-Y'}_F^2
    ~&\geq~
    \frac12\sum_{\ell=1}^m\sum_{z=1}^3p_{\ell,z}^2\enspace.
\end{align*}
Similar to \cref{eq:expected_contraction},
by \cref{clm:su_pair_observation},
we have
\begin{align*}
    \expect{}{
        \norm{X-Y}_F^2-\norm{X'-Y'}_F^2
    }
    ~\geq~
    \frac{3}{4(d-1)}\cdot\norm{X-Y}_F^2\enspace.
\end{align*}
Equivalently,
\begin{equation}
    \expect{}{\norm{X'-Y'}_F^2}
    \leq
    \left(1-\frac{3}{4(d-1)}\right)
    \norm{X-Y}_F^2\enspace.
    \label{eq:su_expected_frobenius_decrease}
\end{equation}
Let $D_U$ denote the Riemannian distance on $\ugroup{d}$, and
let $e^{\mathrm{i}\theta_1},\dots,e^{\mathrm{i}\theta_d}$
be the eigenvalues of $(X')^\dagger Y'$, where $\theta_j\in(-\pi,\pi]$.
Using $|\theta|\leq(\pi/2)|1-e^{\mathrm{i}\theta}|$
for $\theta\in(-\pi,\pi]$ and the fact that the coupling
does not increase the Frobenius distance, we obtain
\[
    \left|\sum_{j=1}^d\theta_j\right|
    \leq\frac{\pi}{2}\sum_{j=1}^d
        |1-e^{\mathrm{i}\theta_j}|
    \leq\frac{\pi d}{2}
        \left(\sum_{j=1}^d
        |1-e^{\mathrm{i}\theta_j}|^2\right)^{1/2}
    = \frac{\pi d}{2}\|X'-Y'\|_F 
    \leq\frac{\pi d}{32(d-1)}<2\pi \enspace.
\]
Together with $\det((X')^\dagger Y')=1$, this gives
$\sum_j\theta_j=0$.
Thus, by \cite[Remark~1.6(a) and Theorem~5.1(a)]{PD25} and
\cref{lem:frobenius_riemannian_distance_comparison},
\[
    D(X',Y')=D_U(X',Y')
    \leq \|X'-Y'\|_F+2\|X'-Y'\|_F^2.
\]
Taking expectations and applying Jensen's inequality, we obtain 
\begin{align*}
    \expect{}{D(X',Y')}
    ~&\leq~
    \sqrt{\expect{}{\norm{X'-Y'}_F^2}}
    +2\expect{}{\norm{X'-Y'}_F^2} \\
    ~&\leq~
    \sqrt{1-\frac{3}{4(d-1)}}\cdot\norm{X-Y}_F
    +2\left(1-\frac{3}{4(d-1)}\right)\norm{X-Y}_F^2 \\
    ~&\leq~
    \left(
        \sqrt{1-\frac{3}{4(d-1)}}
        +\frac{1}{8(d-1)}
    \right)\cdot D(X,Y) \\
    ~&\leq~
    \left(1-\frac{1}{4(d-1)}\right)\cdot D(X,Y)\enspace,
\end{align*}
where the second inequality follows from
\cref{eq:su_expected_frobenius_decrease},
the third inequality uses
\cref{lem:frobenius_riemannian_distance_comparison} and
$D(X,Y)\leq \frac{1}{16(d-1)}$,
and the last inequality uses $\sqrt{1-x}\leq1-x/2$ for $x\in[0,1]$.
\end{proof}

\subsubsection{Projection to the Stiefel Manifold}

We use the contraction estimate on $\su(d)$ to obtain a convergence
bound on $\stframe{d}{k}$ by projecting onto the first $k$ columns.

\begin{theorem}
    \label{thm:high_rank_wasserstein_mixing}
    Let $d\geq2$ be even, $d/2<k\leq d$,
    and $\kackernel$ be the transition kernel of
    the parallel Kac's walk on $\stframe{d}{k}$.
    For every $\varepsilon\in(0,1)$, every integer
    \(
        t\geq
        \left\lceil4(d-1)\log\br{\frac{\pi\sqrt d}{\varepsilon}}\right\rceil
    \)
    and $X,Y\in\stframe{d}{k}$, with $\det X=\det Y$ when $k=d$,
    we have
    \(
        W_D\!\left(\kackernel^t(X,\cdot),\kackernel^t(Y,\cdot)\right)
        \leq\varepsilon
    \).
\end{theorem}

\begin{proof}
Let $\widehat{\kackernel}$ and $\widehat D$ denote the transition kernel
and Riemannian distance on $\su(d)$, respectively.
Choose $U\in\su(d)$ such that $UX=Y$.
The map $V\mapsto VX$ sends the walks on $\su(d)$ started from
$I_d$ and $U$ to the walks on $\stframe{d}{k}$ started from $X$ and $Y$,
respectively.
Since $XX^\dagger$ is an orthogonal projection and satisfies $XX^\dagger\leq I_d$, for every $d\times d$ matrix $A$,
\(
    \norm{AX}_F^2
    =\tr{A^\dagger AXX^\dagger}
    \leq\tr{A^\dagger A}
    =\norm{A}_F^2
\).
The map therefore does not increase lengths of paths or the corresponding
Wasserstein distances.
Thus, iterating \cref{lem:su_parallel_contraction}, we obtain
\begin{align*}
    W_D\!\left(\kackernel^t(X,\cdot),\kackernel^t(Y,\cdot)\right)
    ~\leq~
    W_{\widehat D}\!\left(
        \widehat{\kackernel}^t(I_d,\cdot),\widehat{\kackernel}^t(U,\cdot)
    \right) 
    ~&\leq~
    \left(1-\frac{1}{4(d-1)}\right)^t 
    \cdot\widehat D(I_d,U) \\
    ~&\leq~
    \pi\sqrt d\cdot\e^{-t/(4(d-1))}
    ~\leq~\varepsilon\enspace,
\end{align*}
where the third inequality uses
$\operatorname{diam}_{\widehat D}(\su(d))=\pi\sqrt d$
by \cite[Theorem~5.2(a)]{PD25}.
\end{proof}

\subsection{Wasserstein Mixing Bound}
\label{subsec:wasserstein_mixing_bound}

We are now ready to prove \cref{thm:wasserstein_mixing} using the bounds
established in the preceding two subsections.
We restate the theorem for convenience.

\wassersteinMixingTheorem*

\begin{proof}
Fix $X,Y\in\stframe{d}{k}$, with $\det X=\det Y$ when $k=d$,
$\varepsilon\in(0,1)$, and an integer
$t\geq4000(k+\log d)\log(d/\varepsilon)$.
Suppose first that $k\leq d/2$, and let
\(
    T_0=\left\lceil1000(k+\log d)\log d\right\rceil.
\)
The same argument in the proof of
\cref{thm:high_rank_wasserstein_mixing} also applies here and gives
$D(X,Y)\leq\pi\sqrt d$.
Since $T_0\leq1001(k+\log d)\log d$ and $\pi<d^2$,
\[
    \frac{t}{T_0}
    \geq\frac{4000}{1001}\cdot\frac{\log(d/\varepsilon)}{\log d}
    \geq1+\frac{\log(\pi\sqrt d/\varepsilon)}{\log d}.
\]
Thus, \cref{thm:low_rank_wasserstein_mixing} gives
\[
    W_D\!\left(\kackernel^t(X,\cdot),\kackernel^t(Y,\cdot)\right)
    \leq d^{- \frac{\log(\pi\sqrt d/\varepsilon)  }{\log d}}  \cdot \pi\sqrt d
    =\varepsilon.
\]
Now suppose that $d/2<k\leq d$.
Since $d-1<2k$ and $\pi<d^2$, we have
\[
    4(d-1)\log(\pi\sqrt d/\varepsilon)
    \leq20k\log(d/\varepsilon)
    \leq t.
\]
Therefore, \cref{thm:high_rank_wasserstein_mixing}
directly gives the desired bound.
\end{proof}

\section{Total Variation Mixing}
\label{sec:tv_convergence}
In this section, we prove \cref{thm:tv_mixing}.
The main idea is to first use the Wasserstein bound from \cref{sec:wasserstein_convergence}
to bring two walks close,
and then couple their final states to agree with high probability.
Throughout this section, $d\geq2$ is even and $1\leq k<d$.

In \cref{subsec:tv-notation}, we introduce the transition map
and its Jacobian.
In \cref{subsec:tv-local-smoothing}, we construct a coupling
for nearby initial states under a quantitative surjectivity
assumption on the Jacobian.
We establish the required Jacobian bound in
\cref{subsec:tv-surjectivity}.
Finally, we use these estimates to prove \cref{thm:tv_mixing}
in \cref{subsec:tv-mixing-bound}.

\subsection{The Transition Map and Its Jacobian}
\label{subsec:tv-notation}

Fix even $d\geq2$, $1\leq k<d$, $T\in\N$, and $X\in\stframe{d}{k}$. Let
$\mathcal M=(M_1,\dots,M_T)$ be a sequence of perfect matchings on $[d]$,
and set $m=dT/2$.
We first define a map $\Phi_{X,\mathcal{M}}$ that, for a fixed initial matrix
$X$ and matching sequence $\mathcal{M}$, sends $m$ $\su(2)$ rotations applied
during the parallel Kac's walk to the final matrix after $T$ steps.

\begin{definition}[Transition map]
\label{def:tv-transition-map}
Let
$\mathcal R=(R_{t,e})_{t\in[T],\,e\in M_t}\in\su(2)^m$ be a list of
rotations. For a pair $e=(i,j)$, we regard $R_{t,e}$ as an element of
$\su(d)$ acting on coordinates $i$ and $j$ and as the identity on all other
coordinates. For each $t\in[T]$, set
\(
    G_t(\mathcal R)=\prod_{e\in M_t}R_{t,e}.
\)
The transition map associated with $X$ and $\mathcal M$ is
$\Phi_{X,\mathcal M}:\su(2)^m\longrightarrow\stframe{d}{k}$
defined by
\[
    \Phi_{X,\mathcal M}(\mathcal R)
    =G_T(\mathcal R)\cdots G_1(\mathcal R)X\enspace.
\]
We denote by $\nu_m$ the product Haar probability measure on $\su(2)^m$;
equivalently, under $\mathcal R\sim\nu_m$, the rotations $R_{t,e}$ are
independent and Haar distributed on $\su(2)$.
\end{definition}

To describe how $\Phi_{X,\mathcal{M}}$ responds to perturbations of the rotations,
we write each tangent vector at $R_{t,e}\in\su(2)$ as $R_{t,e}B_{t,e}$ with
$B_{t,e}\in\sualg(2)$, corresponding to the curve
$s\mapsto R_{t,e}\exp(sB_{t,e})$.
The Jacobian map $J_{X,\mathcal{M}}(\mathcal{R})$ defined below sends the
perturbation directions $\mathcal{B}=(B_{t,e})$ to the resulting first-order
change in the final matrix $\Phi_{X,\mathcal{M}}(\mathcal{R})$.

\begin{definition}[Jacobian map]
\label{def:tv-jacobian}
For fixed $\mathcal R=(R_{t,e})_{t\in[T],\,e\in M_t}\in\su(2)^m$,
define the Jacobian map
\(
    J_{X,\mathcal M}(\mathcal R):\sualg(2)^m
    \longrightarrow
    \mathsf T_{\Phi_{X,\mathcal M}(\mathcal R)}\stframe{d}{k}
\)
of $\Phi_{X,\mathcal M}$ at $\mathcal R$ by
\[
    J_{X,\mathcal M}(\mathcal R)[\mathcal B]
    =\left.\frac{\mathrm{d}}{\mathrm{d}s}
    \Phi_{X,\mathcal M}
    \bigl((R_{t,e}\exp(sB_{t,e}))_{t\in[T],\,e\in M_t}\bigr)
    \right|_{s=0}
    \enspace,
\]
where $\mathcal B=(B_{t,e})_{t\in[T],\,e\in M_t}\in\sualg(2)^m$.
\end{definition}

By choosing proper orthonormal bases,
we can view $J_{X,\mathcal M}(\mathcal R)$ as a real matrix
and its transpose can be viewed as a linear map
\(
    J_{X,\mathcal M}(\mathcal R)^\top:
    \mathsf T_{\Phi_{X,\mathcal M}(\mathcal R)}\stframe{d}{k}
    \longrightarrow\sualg(2)^m.
\)
We equip $\sualg(2)^m$ with the norm
\(
    \norm{\mathcal{B}}_F^2=\sum_{t,e}\norm{B_{t,e}}_F^2
\)
for \(\mathcal{B}=(B_{t,e})_{t\in[T],\,e\in M_t}\in\sualg(2)^m\),
and define the \emph{surjectivity modulus} of the Jacobian as
\[
    \sigma_{\mathrm{sur}}\bigl(J_{X,\mathcal{M}}(\mathcal{R})\bigr)
    \coloneqq
    \inf_{\substack{
        H\in\mathsf{T}_{\Phi_{X,\mathcal{M}}(\mathcal{R})}\stframe{d}{k}\\
        \norm{H}_F=1}}
    \norm{J_{X,\mathcal{M}}(\mathcal{R})^{\top}H}_F\enspace.
\]
Intuitively, a lower bound on $\sigma_{\mathrm{sur}}(J)$ ensures
that any change in the final state can be produced by a controlled adjustment of the rotations.

\subsection{Coupling Nearby Initial States}
\label{subsec:tv-local-smoothing}

This subsection constructs a coupling for nearby initial states
under a lower bound on $\sigma_{\mathrm{sur}}$.
The construction adjusts the rotations so that the final states agree,
while controlling the total variation distance
between the adjusted rotation distribution and ideal Haar measure.

\begin{lemma}
\label{lem:tv-gate-correction}
Fix $T\in\N$, $X\in\stframe{d}{k}$ and a matching sequence $\mathcal{M}$ of length $T$.
Let $m = dT/2$ and $\nu_m$ be the product Haar probability on $\su(2)^m$.
Suppose that a set $\mathfrak{R}_{\mathcal M}\subseteq\su(2)^m$ satisfies that
$\Pr_{\mathcal{R} \sim \nu_m}[ \mathcal{R} \in \mathfrak{R}_{\mathcal M} ] \ge 1 - \eta_{\mathcal M}$ and
that, for some $\kappa\ge1$,
$\sigma_{\mathrm{sur}}\bigl(J_{X,\mathcal{M}}(\mathcal{R})\bigr)\geq\kappa^{-1}$
for every $\mathcal{R}\in\mathfrak{R}_{\mathcal M}$.
Set $C=2^{20}$. Then, for every $Y\in\stframe{d}{k}$ with
$\|X-Y\|_F\le 1/(C\kappa^4d^9T^3)$,
there is a map
$\mathcal T:\su(2)^m\to\su(2)^m$
such that the following two properties hold:
\begin{enumerate}[label=\textup{(\roman*)}]
  \item For every $\mathcal{R}\in\mathfrak{R}_{\mathcal M}$,
  \(
    \Phi_{Y,\mathcal{M}}(\mathcal T(\mathcal{R}))=\Phi_{X,\mathcal{M}}(\mathcal{R})
  \).
  \item For $\mathcal{R}\sim\nu_m$, the total variation distance between
  the distribution of $\mathcal T(\mathcal{R})$ and $\nu_m$
  is at most $C\kappa^4d^9T^3\|X-Y\|_F$.
\end{enumerate}
Thus, there is a coupling of rotation lists $\mathcal{R}_X$ and $\mathcal{R}_Y$,
each with distribution $\nu_m$, such that the final matrices
$X_T=\Phi_{X,\mathcal{M}}(\mathcal{R}_X)$ and $Y_T=\Phi_{Y,\mathcal{M}}(\mathcal{R}_Y)$ satisfy
\(
  \prob{}{X_T\ne Y_T}\le\eta_{\mathcal M}+C\kappa^4d^9T^3\|X-Y\|_F.
\)
\end{lemma}

\begin{proof}
We first construct the map $\mathcal T$ and then verify
the two properties in the statement.
The main idea is to move the initial state from $X$ to $Y$
along a path while adjusting the rotations to keep the final
state fixed for rotation lists in $\mathfrak R_{\mathcal M}$.

We connect the initial matrices $X$ and $Y$ by the path
$Z_q=(1-q)X+qY$ for $q\in[0,1]$.
Note that
\(
	Z_q^\dagger Z_q  =  I_k - q(1-q)(X-Y)^\dagger(X-Y)  \succeq\bigl(1-\|X-Y\|_F^2/4\bigr)I_k
\).
By assumption, $\|X-Y\|_F\le 1/4$, so $Z_q^\dagger Z_q$ is invertible.
Then, we define $X_q=Z_q(Z_q^\dagger Z_q)^{-1/2}$ such that $X_q\in \stframe{d}{k}$,
and write
\begin{itemize}
  \item $\Phi_{X_q,\mathcal M}: \su(2)^m\to\stframe{d}{k}$ for the map from rotations to the final matrix after $T$ steps of the walk starting at $X_q$ as in \cref{def:tv-transition-map},
  \item $J_{X_q,\mathcal M}(\mathcal R) : \sualg(2)^m\to\mathsf T_{\Phi_{X_q,\mathcal M}(\mathcal R)}\stframe{d}{k}$ for the Jacobian of the map $\Phi_{X_q,\mathcal M}$ at $\mathcal R \in \su(2)^m$ as in \cref{def:tv-jacobian}, and
  \item $H_{\mathcal R}(q)\coloneqq
    \frac{\mathrm{d}}{\mathrm{d}q}\Phi_{X_q,\mathcal M}(\mathcal R)$
    for the derivative of the final matrix along the path $X_q$,
    with fixed $\mathcal R\in\su(2)^m$.
\end{itemize}

The map $J_{X_q,\mathcal M}(\mathcal R)^\top$ is defined on
$\mathsf T_{\Phi_{X_q,\mathcal M}(\mathcal R)}\stframe{d}{k}$,
which is a subspace of $\C^{d\times k}$.
Below, we extend it to a map
$\C^{d\times k}\to\sualg(2)^m$ by setting it to zero
on the orthogonal complement of this tangent space.
Let $\Pi_{\Phi_{X_q,\mathcal M}(\mathcal R)}$ denote the orthogonal
projection onto $\mathsf T_{\Phi_{X_q,\mathcal M}(\mathcal R)}\stframe{d}{k}$ and define
\[
  A_{q,\mathcal R}
  \coloneqq J_{X_q,\mathcal M}(\mathcal R)
            J_{X_q,\mathcal M}(\mathcal R)^\top
            +I-\Pi_{\Phi_{X_q,\mathcal M}(\mathcal R)}.
\]
Thus, $A_{q,\mathcal R}$ acts as
$J_{X_q,\mathcal M}(\mathcal R)J_{X_q,\mathcal M}(\mathcal R)^\top$
on $\mathsf T_{\Phi_{X_q,\mathcal M}(\mathcal R)}\stframe{d}{k}$
and as the identity on its orthogonal complement.
Define
\[
  f(\mathcal R) \coloneqq \operatorname{tr}
       \bigl((A_{0,\mathcal R}+(32kd\kappa^2)^{-1}I)^{-1}\bigr) 
  \quad
  \text{and}
  \quad \Omega \coloneqq \{\mathcal R:f(\mathcal R)<12kd\kappa^2\} \enspace.
\]
It is easy to see that $A_{0,\mathcal R}$ is self-adjoint and positive
semidefinite, and thus $f(\mathcal R)$ is well defined for every
$\mathcal R\in\su(2)^m$.
Intuitively, $f(\mathcal R)$ measures how difficult it is to adjust
the final matrix by changing the rotations.
When $f$ is small,
every tangent direction can be produced with a controlled rotation adjustment,
so we use $f$ to decide where to activate the adjustment.
The following claim  provides the bounds needed later,
which is proved in \cref{app:tv_correction_invertibility_proof}.

\begin{restatable}{claim}{correctionInvertibilityClaim}
\label{clm:tv_correction_invertibility}
For every $\mathcal R\in\mathfrak{R}_{\mathcal M}$, $f(\mathcal R)<2kd\kappa^2$.
Moreover, for every $\mathcal R\in\Omega$ and $q\in[0,1]$,
$A_{q,\mathcal R}$ is invertible on $\C^{d\times k}$
with $\|A_{q,\mathcal R}^{-1}\|_{\mathrm{op}}\le40kd\kappa^2$.
\end{restatable}
For $\mathcal R\in\Omega$, let $Q_q(\mathcal R)=J_{X_q,\mathcal M}(\mathcal R)^\top A_{q,\mathcal R}^{-1}$.
Choose a continuously differentiable function $\chi:\mathbb R\to[0,1]$
such that $\chi(x)=1$ for $x\le4kd\kappa^2$,
$\chi(x)=0$ for $x\ge8kd\kappa^2$, and
$|\chi'(x)|\le1/(2kd\kappa^2)$.\footnote{For example, with
$u=(x-4kd\kappa^2)/(4kd\kappa^2)$, take
$\chi(x)=1$ for $u\le0$, $\chi(x)=1-3u^2+2u^3$ for $0<u<1$,
and $\chi(x)=0$ for $u\ge1$.}
Define the rotation directions
$\mathcal B^{q,\mathcal R}=(B_{t,e}^{q,\mathcal R})_{t\in[T],\,e\in M_t}
\in\sualg(2)^m$ by
\[
  \mathcal B^{q,\mathcal R}=
  \begin{cases}
    -\chi(f(\mathcal R))Q_q(\mathcal R)H_{\mathcal R}(q),
      & \mathcal R\in\Omega,\\
    0, & \mathcal R\notin\Omega.
  \end{cases}
\]

\smallskip
\noindent\textbf{Construction of the map $\mathcal T_q$.}
For each initial list $\mathcal R=(R_{t,e})$, let
$\mathcal R(q)=(R_{t,e}(q))$ solve
\[
  \frac{\mathrm d}{\mathrm d q}R_{t,e}(q)
  =R_{t,e}(q)B_{t,e}^{q,\mathcal R(q)},
  \qquad R_{t,e}(0)=R_{t,e}
  \quad(t\in[T],\ e\in M_t).
\]
Define $\mathcal T_q(\mathcal R)=\mathcal R(q)$ and
$\mathcal T=\mathcal T_1$.
We have the following properties of the map $\mathcal T_q$ and its proof is given in
\cref{app:tv_correction_flow_bounds_proof}.

\begin{restatable}{claim}{correctionFlowBoundsClaim}
\label{clm:tv_correction_flow_bounds}
The map $\mathcal T_q$ is well defined for $q\in[0,1]$,
and each $\mathcal T_q$ is continuously differentiable.
Moreover, for every $q\in[0,1]$:
\begin{enumerate}[label=\textup{(\alph*)}]
\item For every $\mathcal R\in\su(2)^m$, each rotation $R_{t,e}(q)$
remains in $\su(2)$, so $\mathcal T_q(\mathcal R)\in\su(2)^m$.
\item For every $\mathcal R\in\su(2)^m$,
$|f(\mathcal T_q(\mathcal R))-f(\mathcal R)|\le2kd\kappa^2$.
\item If $\mathcal R\sim\nu_m$, the total variation distance between
the distribution of $\mathcal T_q(\mathcal R)$ and $\nu_m$
is at most $C\kappa^4d^9T^3\|X-Y\|_F$.
\end{enumerate}
\end{restatable}

Note that \textup{(ii)} follows immediately from
\cref{clm:tv_correction_flow_bounds}\textup{(c)} with $q=1$.
We now verify \textup{(i)}.
Consider a fixed $\mathcal R\in\mathfrak{R}_{\mathcal M}$.
Combining \cref{clm:tv_correction_invertibility} with
\cref{clm:tv_correction_flow_bounds}\textup{(b)},
we have
$f(\mathcal T_q(\mathcal R))<4kd\kappa^2$ and
therefore $\chi(f(\mathcal T_q(\mathcal R)))=1$ for all $q\in[0,1]$.
As $q$ varies, both the initial matrix $X_q$ and the rotations
$\mathcal T_q(\mathcal R)$ change. The chain rule gives
\begin{align*}
  \frac{\mathrm d}{\mathrm d q}\Phi_{X_q,\mathcal M}(\mathcal T_q(\mathcal R))
  &=\left.\frac{\mathrm d}{\mathrm d u}
      \Phi_{X_u,\mathcal M}(\mathcal T_q(\mathcal R))\right|_{u=q}
    +\left.\frac{\mathrm d}{\mathrm d u}
      \Phi_{X_q,\mathcal M}(\mathcal T_u(\mathcal R))\right|_{u=q}.
\end{align*}
The first term varies the initial matrix while keeping the rotations
fixed, so it equals $H_{\mathcal T_q(\mathcal R)}(q)$ by definition.
The second term keeps the initial matrix fixed and varies the rotations.
By the definition of $\mathcal T_q$, we have
$\frac{\mathrm d}{\mathrm d q}R_{t,e}(q)
=R_{t,e}(q)B_{t,e}^{q,\mathcal T_q(\mathcal R)}$.
Also, note that
\(
  \left.\frac{\mathrm d}{\mathrm d s}
    \bigl(R_{t,e}(q)\exp(sB_{t,e}^{q,\mathcal T_q(\mathcal R)})\bigr)
  \right|_{s=0}
  =R_{t,e}(q)B_{t,e}^{q,\mathcal T_q(\mathcal R)}.
\)
By the product rule and the definition of $J$, we obtain
\[
  \left.\frac{\mathrm d}{\mathrm d u}
    \Phi_{X_q,\mathcal M}(\mathcal T_u(\mathcal R))\right|_{u=q}
  =J_{X_q,\mathcal M}(\mathcal T_q(\mathcal R))
    [\mathcal B^{q,\mathcal T_q(\mathcal R)}].
\]
Since $\chi(f(\mathcal T_q(\mathcal R)))=1$, we have
$\mathcal B^{q,\mathcal T_q(\mathcal R)}=-Q_q(\mathcal T_q(\mathcal R))H_{\mathcal T_q(\mathcal R)}(q)$.
Moreover, $H_{\mathcal T_q(\mathcal R)}(q)$ lies in the tangent space at
$\Phi_{X_q,\mathcal M}(\mathcal T_q(\mathcal R))$, on which $Q_q(\mathcal T_q(\mathcal R))$
is a right inverse of $J_{X_q,\mathcal M}(\mathcal T_q(\mathcal R))$.
Combining the two derivative terms therefore yields
\begin{align*}
  \frac{\mathrm d}{\mathrm d q}\Phi_{X_q,\mathcal M}(\mathcal T_q(\mathcal R))
  &=H_{\mathcal T_q(\mathcal R)}(q)
    +J_{X_q,\mathcal M}(\mathcal T_q(\mathcal R))[\mathcal B^{q,\mathcal T_q(\mathcal R)}]\\
  &=H_{\mathcal T_q(\mathcal R)}(q)
    -J_{X_q,\mathcal M}(\mathcal T_q(\mathcal R))Q_q(\mathcal T_q(\mathcal R))H_{\mathcal T_q(\mathcal R)}(q)
    =0.
\end{align*}
Thus, $\Phi_{X_q,\mathcal M}(\mathcal T_q(\mathcal R))$
is independent of $q$.
Evaluating at $q=0$ and $q=1$ proves
$\Phi_{Y,\mathcal M}(\mathcal T(\mathcal R))
=\Phi_{X,\mathcal M}(\mathcal R)$ for every
$\mathcal R\in\mathfrak{R}_{\mathcal M}$, as required in \textup{(i)}.

Finally, we prove the coupling bound.
Sample $\mathcal R_X\sim\nu_m$ and set
$\widetilde{\mathcal R}=\mathcal T(\mathcal R_X)$.
By property \textup{(ii)}, we can couple $\widetilde{\mathcal R}$
with $\mathcal R_Y\sim\nu_m$ such that
\[
  \Pr[\widetilde{\mathcal R}\ne\mathcal R_Y]
  \le C\kappa^4d^9T^3\|X-Y\|_F.
\]
If $\mathcal R_X\in\mathfrak{R}_{\mathcal M}$ and
$\widetilde{\mathcal R}=\mathcal R_Y$, property \textup{(i)} gives
\[
  X_T=\Phi_{X,\mathcal M}(\mathcal R_X)
     =\Phi_{Y,\mathcal M}(\mathcal T(\mathcal R_X))
     =\Phi_{Y,\mathcal M}(\mathcal R_Y)=Y_T.
\]
Therefore, a union bound gives
\[
  \Pr[X_T\ne Y_T]
  \le\Pr[\mathcal R_X\notin\mathfrak{R}_{\mathcal M}]
     +\Pr[\widetilde{\mathcal R}\ne\mathcal R_Y]
  \le\eta_{\mathcal M}+C\kappa^4d^9T^3\|X-Y\|_F,
\]
which is the claimed coupling bound.
\end{proof}

\subsection{Quantitative Surjectivity of the Jacobian}
\label{subsec:tv-surjectivity}

We establish the quantitative surjectivity bound needed
in \cref{lem:tv-gate-correction}.
For each initial state, the following lemma bounds
$\sigma_{\mathrm{sur}}$ from below with high probability
over the matchings and rotations sampled during
$O((k+\log d)\log d)$ steps.

\begin{lemma}
    \label{lem:tv_endpoint_control}
    Let $1\leq k<d$ and
    $T=\lceil2^{16}(k+\log d)\log d\rceil$.
    For every $X\in\stframe{d}{k}$, we have
    \[
        \prob{\mathcal M,\mathcal R}{
            \sigma_{\mathrm{sur}}(J_{X,\mathcal M}(\mathcal R))
            <\frac{\sqrt{k+\log d}}{2d}}
        \leq d^{-30}\enspace,
    \]
    where $\mathcal M$ and $\mathcal R$ are the matchings and Haar
    rotations sampled in $T$ steps of the walk.
\end{lemma}

\begin{proof}
By the definition of $\sigma_{\mathrm{sur}}$, it suffices to show that,
$\norm{J_{X,\mathcal M}(\mathcal R)^\top V}_F
\geq \sqrt{k+\log d}/(2d)$ holds for every unit vector
$V\in\mathsf T_{\Phi_{X,\mathcal M}(\mathcal R)}\stframe{d}{k}$,
with probability at least $1-d^{-30}$.

Fix $X\in\stframe{d}{k}$.
For each selection $(\mathcal M,\mathcal R)$ of the walk, define
$X_0=X$ and $X_t=G_tX_{t-1}$ for $1\leq t\leq T$,
with $G_t$ as in \cref{def:tv-transition-map}.
We will specify below a rule for choosing rotation directions
and use it to construct a sequence of tangent vectors as follows.
Given $H\in\mathsf T_X\stframe{d}{k}$, set $H_0=H$.
For $t=0,\ldots,T-1$, given $H_t\in\mathsf T_{X_t}\stframe{d}{k}$,
the rule chooses $B_{t+1,e}\in\sualg(2)$ for each $e\in M_{t+1}$
using $X_t$, $H_t$, and $M_{t+1}$, with linear dependence on $H_t$.
Let $B_{t+1}\in\sualg(d)$ be the matrix whose blocks on the pairs
of $M_{t+1}$ are $B_{t+1,e}$ and whose other entries are zero, and set
\begin{equation}
    H_{t+1}=G_{t+1}(H_t+B_{t+1}X_t)\enspace.
    \label{eq:tv_tangent_recursion}
\end{equation}
This recursion gives
$H_{t+1}\in\mathsf T_{X_{t+1}}\stframe{d}{k}$.
Write $\mathcal B=(B_{t,e})_{t\in[T],\,e\in M_t}$ for the resulting
list of rotation directions.
For each fixed $(\mathcal M,\mathcal R)$, both $\mathcal B$ and $H_T$
depend linearly on the initial $H$.

We will show that, for every fixed
$H\in\mathsf T_X\stframe{d}{k}$,
\begin{equation}
    \norm{\mathcal B}_F^2\leq\frac{d^2}{k+\log d}\norm{H}_F^2
    \enspace,\qquad
    \expect{\mathcal M,\mathcal R}{\norm{H_T}_F^2}
        \leq d^{-37}\norm{H}_F^2\enspace,
    \label{eq:tv_control_targets}
\end{equation}
where the first bound will hold for every $(\mathcal M,\mathcal R)$.

We first prove the lemma assuming these two bounds.
Note that the map $L:H\mapsto H_T$ is linear, and
$\mathsf T_X\stframe{d}{k}$ has dimension $D=2dk-k^2\leq d^2$.
By Markov's inequality and a union bound over an orthonormal basis
$E_1,\ldots,E_D$ of $\mathsf T_X\stframe{d}{k}$, the second bound in \cref{eq:tv_control_targets}
implies that $\norm{L(E_i)}_F\leq1/(2\sqrt D)$ for every $i$,
except with probability at most $4D^2d^{-37}\leq d^{-30}$.
Therefore, with probability at least $1-d^{-30}$ over the choice of
$(\mathcal M,\mathcal R)$, we have that, for every
$H=\sum_{i=1}^D a_iE_i\in\mathsf T_X\stframe{d}{k}$ with $a_i\in\mathbb R$,
\begin{align}
    \norm{H_T}_F =\norm{\sum_{i=1}^D a_iL(E_i)}_F\!\!
    \leq\sum_{i=1}^D |a_i|\norm{L(E_i)}_F \leq\frac1{2\sqrt D}\sum_{i=1}^D |a_i|
    \leq\frac12\left(\sum_{i=1}^D |a_i|^2\right)^{1/2}
    \!\!=\frac12\norm{H}_F\,. \label{eq:tv_endpoint_control_HT_bound}
\end{align}
Fix a $(\mathcal M,\mathcal R)$ for which \cref{eq:tv_endpoint_control_HT_bound} holds.
For any $V\in\mathsf T_{X_T}\stframe{d}{k}$ with $\norm{V}_F=1$,
take $H=G_1^\dagger\cdots G_T^\dagger V$.
Then $H\in\mathsf T_X\stframe{d}{k}$, $\norm{H}_F=1$, and
$G_T\cdots G_1H=V$.
Iterating \cref{eq:tv_tangent_recursion}, we have
\[
    H_T=G_T\cdots G_1H +\sum_{t=1}^T G_T\cdots G_tB_tX_{t-1}\enspace.
\]
To identify the sum, hold the chosen directions $\mathcal B$ fixed and set
$R_{t,e}(s)=R_{t,e}\exp(sB_{t,e})$ and
$G_t(s)=\prod_{e\in M_t}R_{t,e}(s)$.
Since the pairs in $M_t$ are disjoint,
$G_t(s)=G_t\exp(sB_t)$, so $G_t'(0)=G_tB_t$.
By \cref{def:tv-jacobian} and the product rule,
\begin{align*}
    J_{X,\mathcal M}(\mathcal R)[\mathcal B]
    &=\left.\frac{\mathrm d}{\mathrm ds}
        G_T(s)\cdots G_1(s)X\right|_{s=0} =\sum_{t=1}^T
        G_T\cdots G_{t+1}(G_tB_t)G_{t-1}\cdots G_1X\\
    &=\sum_{t=1}^T G_T\cdots G_tB_tX_{t-1}\enspace.
\end{align*}
Combining these identities yields
\begin{equation*}
    H_T=G_T\cdots G_1H+J_{X,\mathcal M}(\mathcal R)[\mathcal B]\enspace.
\end{equation*}
For our choice of $H$, we have
$H_T=V+J_{X,\mathcal M}(\mathcal R)[\mathcal B]$.
Therefore,
\begin{align*}
    1~=~\norm{V}_F^2
    &~=~\langle V,H_T\rangle_F
        -\langle V,J_{X,\mathcal M}(\mathcal R)[\mathcal B]\rangle_F
    ~=~\langle V,H_T\rangle_F
        -\langle J_{X,\mathcal M}(\mathcal R)^\top V,\mathcal B\rangle_F\\
    &~\leq~|\langle V,H_T\rangle_F|
        +|\langle J_{X,\mathcal M}(\mathcal R)^\top V,\mathcal B\rangle_F|
        ~\leq~\norm{V}_F\norm{H_T}_F
        +\norm{J_{X,\mathcal M}(\mathcal R)^\top V}_F\norm{\mathcal B}_F\\
    &~\leq~\frac12+\frac d{\sqrt{k+\log d}}
        \norm{J_{X,\mathcal M}(\mathcal R)^\top V}_F\enspace,
\end{align*}
where we used the Cauchy--Schwarz inequality in the second inequality and the bounds in \cref{eq:tv_control_targets,eq:tv_endpoint_control_HT_bound} in the last inequality.
Rearranging gives
$\norm{J_{X,\mathcal M}(\mathcal R)^\top V}_F
\geq\sqrt{k+\log d}/(2d)$.

We now construct the rotation directions $\mathcal B$ and prove
\cref{eq:tv_control_targets}.
For any $X\in\stframe{d}{k}$ and $H\in\mathsf T_X\stframe{d}{k}$,
we define $A_X(H)$ to be the matrix of smallest Frobenius norm among all
$A\in\sualg(d)$ satisfying $AX=H$.
We have the following properties of $A_X$, which is proved in \cref{app:tv_lift}.

\begin{restatable}{claim}{tvLiftClaim}
\label{clm:tv_lift}
For every $X\in\stframe{d}{k}$, the matrix $A_X(H)$ exists and is unique
for every $H\in\mathsf T_X\stframe{d}{k}$, and depends linearly on $H$.
Moreover, the following properties hold.
\begin{enumerate}
    \item For every $H\in\mathsf T_X\stframe{d}{k}$,
    \(
        \norm{H}_F^2
        \leq\norm{A_X(H)}_F^2
        \leq2d\norm{H}_F^2.
    \)
    \item For every unitary $G\in\mathbb C^{d\times d}$,
    \(
        \norm{A_{GX}(GH)}_F
        =\norm{A_X(H)}_F.
    \)

    \item Suppose every row of $X$ has squared Frobenius norm at most $\Delta>0$
    and $B\in\sualg(d)$ is block diagonal on a perfect matching,
    with blocks in $\sualg(2)$.
    Then, for every $H\in\mathsf T_X\stframe{d}{k}$,
    \begin{equation}
        \norm{A_X(H+BX)}_F^2
        \leq\norm{A_X(H)}_F^2
            +2\langle A_X(H),B\rangle_F
            +5\Delta\norm{B}_F^2\enspace.
        \label{eq:tv_lift_change}
    \end{equation}
\end{enumerate}
\end{restatable}

\paragraph{Construction of $\mathcal B$.}
Set $t_0=\lceil30\log(2d^{41})\rceil$ and
$\Delta=4(k+41\log d)/d$.
For $t\geq t_0$, define the event
$\mathcal G_t=\{\max_{i\in[d]}\norm{\mathbf e_i^\dagger X_t}_2^2\leq\Delta\}$.
By \cref{lem:maximum_row_norm} with $\delta=d^{-40}$,
\begin{equation}
    \prob{}{\mathcal G_t^{\mathrm c}}\leq d^{-40}
    \qquad(t\geq t_0)\enspace.
    \label{eq:tv_row_norm_failure}
\end{equation}
For each $t=0,\ldots,T-1$, we choose the rotation directions as follows.
\begin{enumerate}
    \item If $t<t_0$, or if $t\geq t_0$ and $\mathcal G_t$ does not occur,
    set $B_{t+1,e}=0$ for every $e\in M_{t+1}$, so $B_{t+1}=0$.

    \item If $t\geq t_0$ and $\mathcal G_t$ occurs, set
    $A=A_{X_t}(H_t)$.
    For each pair $e=(i,j)$, let $A_e$ be the principal
    $2\times2$ block of $A$ indexed by $e$, and write
    $A_e^0=A_e-\tfrac12\tr{A_e}I_2$.
    For each $e\in M_{t+1}$, set
    \(
        B_{t+1,e}=-\frac{A_e^0}{5\Delta}.
    \)
\end{enumerate}
The choice of case is independent
of $H$, and each $B_{t+1,e}$ depends linearly on $H_t$.
Together with \cref{eq:tv_tangent_recursion}, this shows that
$\mathcal B$ and $H_T$ depend linearly on the initial vector $H$.

\paragraph{Verifying \cref{eq:tv_control_targets}.}
Define $W_t=\norm{A_{X_t}(H_t)}_F^2$ for $t=0,\ldots,T$,
and fix $t_0\leq t<T$.
By \cref{eq:tv_tangent_recursion} and the unitary invariance in
\cref{clm:tv_lift},
\begin{align*}
    W_{t+1}
    &=\norm{A_{G_{t+1}X_t}
        \bigl(G_{t+1}(H_t+B_{t+1}X_t)\bigr)}_F^2 =\norm{A_{X_t}(H_t+B_{t+1}X_t)}_F^2\enspace.
\end{align*}
On $\mathcal G_t$, every row of $X_t$ has squared Frobenius norm at most $\Delta$,
and $B_{t+1}$ has blocks in $\sualg(2)$ on the matching $M_{t+1}$.
Thus, applying \cref{eq:tv_lift_change} with
$(X,H,B)=(X_t,H_t,B_{t+1})$ and writing $A=A_{X_t}(H_t)$ gives
\[
    W_{t+1}\leq W_t+2\langle A,B_{t+1}\rangle_F
        +5\Delta\norm{B_{t+1}}_F^2\enspace.
\]
Rearranging and substituting $B_{t+1,e}=-A_e^0/(5\Delta)$, we obtain
\begin{align}
    W_t-W_{t+1}
    &\geq-2\langle A,B_{t+1}\rangle_F
        -5\Delta\norm{B_{t+1}}_F^2 =\frac1{5\Delta}
        \sum_{e\in M_{t+1}}\norm{A_e^0}_F^2
     =5\Delta\sum_{e\in M_{t+1}}\norm{B_{t+1,e}}_F^2\enspace,
    \label{eq:tv_energy_decrease}
\end{align}
where we used $\langle A_e,A_e^0\rangle_F=\norm{A_e^0}_F^2$.
When no correction is made, $B_{t+1}=0$ and unitary invariance gives
$W_{t+1}=W_t$.
Thus, for every $t=0,\ldots,T-1$,
\(
  W_t-W_{t+1} \geq0.
\)
In particular,
\begin{equation}
    0\leq W_t\leq W_0
    \qquad(t=0,\ldots,T)\enspace.
    \label{eq:tv_energy_bound}
\end{equation}
Summing over all steps and using $W_T\geq0$, we obtain
\[
    5\Delta\norm{\mathcal B}_F^2
    =5\Delta\sum_{t=0}^{T-1}\sum_{e\in M_{t+1}}
        \norm{B_{t+1,e}}_F^2
    \leq\sum_{t=0}^{T-1}(W_t-W_{t+1})
    =W_0-W_T\leq W_0\enspace.
\]
Therefore, by \cref{clm:tv_lift}, we have
\[
    \norm{\mathcal B}_F^2
    \leq\frac{W_0}{5\Delta}
    \leq\frac{2d}{5\Delta}\norm{H}_F^2
    \leq\frac{d^2}{k+\log d}\norm{H}_F^2\enspace,
\]
which proves the first bound in \cref{eq:tv_control_targets}.

We next prove the second bound in \cref{eq:tv_control_targets}.
Fix $t_0\leq t<T$ and write $A=A_{X_t}(H_t)$.
Since $A\in\sualg(d)$, we have $\tr{A}=0$ and hence
\[
    \sum_{i<j}|A_{ii}-A_{jj}|^2
    =d\sum_{i=1}^d|A_{ii}|^2-|\tr{A}|^2
    =d\sum_{i=1}^d|A_{ii}|^2\enspace.
\]
By the definition of $A_{(i,j)}^0$, it follows that
\begin{align*}
    \sum_{i<j}\norm{A_{(i,j)}^0}_F^2
    &=\sum_{i\ne j}|A_{ij}|^2
        +\frac12\sum_{i<j}|A_{ii}-A_{jj}|^2 =\sum_{i\ne j}|A_{ij}|^2
        +\frac d2\sum_{i=1}^d|A_{ii}|^2
    \geq\norm{A}_F^2=W_t\enspace,
\end{align*}
where the inequality uses $d\geq2$.

Set $\gamma=1/(5\Delta(d-1))$.
Conditioned on $X_t,H_t$, the matching $M_{t+1}$ is uniform,
and each pair belongs to $M_{t+1}$ with probability $1/(d-1)$.
Thus, on $\mathcal G_t$, \cref{eq:tv_energy_decrease} gives
\begin{align*}
    \expect{}{\left.W_{t+1}\,\right\vert X_t,H_t}
    &\leq W_t-\frac1{5\Delta}
        \expect{}{\left.\sum_{e\in M_{t+1}}\norm{A_e^0}_F^2
            \,\right\vert X_t,H_t}\\
    &=W_t-\frac1{5\Delta(d-1)}
        \sum_{i<j}\norm{A_{(i,j)}^0}_F^2
    \leq(1-\gamma)W_t\enspace.
\end{align*}
On $\mathcal G_t^{\mathrm c}$, we have $B_{t+1}=0$ and $W_{t+1}=W_t$.
Taking expectations over both events, we obtain
\begin{align}
    \expect{}{W_{t+1}}
    &\leq(1-\gamma)\expect{}{\mathbf{1}_{\mathcal G_t}W_t}
        +\expect{}{\mathbf{1}_{\mathcal G_t^{\mathrm c}}W_t}\nonumber\\
    &=(1-\gamma)\expect{}{W_t}
        +\gamma\expect{}{\mathbf{1}_{\mathcal G_t^{\mathrm c}}W_t}\nonumber\\
    &\leq(1-\gamma)\expect{}{W_t}+\gamma d^{-40}W_0\enspace,
    \label{eq:tv_energy_expectation_recursion}
\end{align}
where the last inequality follows from
\cref{eq:tv_energy_bound,eq:tv_row_norm_failure}.
No correction is made before time $t_0$, so $W_{t_0}=W_0$.
Iterating \cref{eq:tv_energy_expectation_recursion} yields
\begin{align*}
    \expect{}{W_T}
    &\leq(1-\gamma)^{T-t_0}W_0
        +\gamma d^{-40}W_0\sum_{j=0}^{T-t_0-1}(1-\gamma)^j\\
    &\leq\bigl((1-\gamma)^{T-t_0}+d^{-40}\bigr)W_0\enspace,
\end{align*}
where we used $0<\gamma<1$ and the geometric series bound.
By the choices of $\Delta$ and $t_0$,
\[
    \gamma=\frac{d}{20(k+41\log d)(d-1)}
    \geq\frac1{820(k+\log d)} \quad\text{and}\quad
    t_0\leq1262\log d\enspace.
\]
Our choice of $T$ therefore gives
\[
    \gamma(T-t_0)\geq\frac{T-t_0}{820(k+\log d)}
    \geq78\log d\enspace.
\]
Consequently, using $1-\gamma\leq\e^{-\gamma}$ and the norm bounds
in \cref{clm:tv_lift}, we have
\begin{align*}
    \expect{}{\norm{H_T}_F^2}
    &\leq\expect{}{W_T}
    \leq\bigl(\e^{-\gamma(T-t_0)}+d^{-40}\bigr)W_0 \leq(d^{-78}+d^{-40})W_0
    \leq d^{-39}W_0
    \leq d^{-37}\norm{H}_F^2\enspace,
\end{align*}
where we used $d\geq2$ in the last two inequalities.
This proves the second bound in \cref{eq:tv_control_targets}
and completes the proof.
\end{proof}

\subsection{Total Variation Mixing Bound}
\label{subsec:tv-mixing-bound}

This subsection proves \cref{thm:tv_mixing}.
We restate the theorem for convenience.

\tvMixingTheorem*

\begin{proof}
Set $T_1=\lceil564000(k+\log d)\log d\rceil$,
$T_2=\lceil2^{16}(k+\log d)\log d\rceil$, and $T_3=T_1+T_2$.
We first show that, for every $X,Y\in\stframe{d}{k}$, there is a
coupling of the walks started from $X$ and $Y$ such that
\begin{equation}
    \prob{}{X_{T_3}\ne Y_{T_3}}\leq d^{-20}\enspace.
    \label{eq:tv_block_coupling}
\end{equation}
We then apply this coupling repeatedly to obtain the claimed mixing bound.

\paragraph{Bringing the two walks close.}
Fix $X,Y\in\stframe{d}{k}$.
By \cref{thm:wasserstein_mixing} with $\varepsilon=d^{-140}$,
we have
\[
    W_D\!\left(\kackernel^{T_1}(X,\cdot),
        \kackernel^{T_1}(Y,\cdot)\right)
    \leq d^{-140}\enspace.
\]
Thus, by \cref{lem:frobenius_riemannian_distance_comparison},
we can couple the states $X',Y'$ after $T_1$ steps so that
\[
    \expect{}{\norm{X'-Y'}_F}
    \leq\expect{}{D(X',Y')}
    \leq d^{-140}\enspace.
\]
By Markov's inequality,
\begin{equation}
    \prob{}{\norm{X'-Y'}_F>d^{-116}}
    \leq d^{116}\cdot\expect{}{\norm{X'-Y'}_F}
    \leq d^{-24}\enspace.
    \label{eq:tv_close_states_failure}
\end{equation}

\paragraph{Coupling the final matrices.}
Fix $X',Y'$ satisfying $\norm{X'-Y'}_F\leq d^{-116}$.
Set $\kappa=2d/\sqrt{k+\log d}$ and let $C=2^{20}$ be the constant
in \cref{lem:tv-gate-correction}.
We have
\begin{equation}
    C\kappa^4d^9T_2^3\norm{X'-Y'}_F
    \leq2^{20}(2d)^4d^9(2^{17}d^2)^3d^{-116}
    =2^{75}d^{-97}\leq d^{-22}\enspace.
    \label{eq:tv_correction_error_bound}
\end{equation}
In particular, $\norm{X'-Y'}_F\leq1/(C\kappa^4d^9T_2^3)$,
as required in \cref{lem:tv-gate-correction}.

For each matching sequence $\mathcal M$, define
\[
    \mathfrak R_{\mathcal M}
    =\left\{\mathcal R\in\su(2)^m:
        \sigma_{\mathrm{sur}}(J_{X',\mathcal M}(\mathcal R))
        \geq\kappa^{-1}\right\},
    \qquad
    \eta_{\mathcal M}
    =\prob{\mathcal R\sim\nu_m}{\mathcal R\notin\mathfrak R_{\mathcal M}},
\]
where $m=dT_2/2$.
Since $\kappa^{-1}=\sqrt{k+\log d}/(2d)$,
\cref{lem:tv_endpoint_control} gives
\begin{equation}
    \expect{\mathcal M}{\eta_{\mathcal M}}
    =\prob{\mathcal M,\mathcal R}{
        \sigma_{\mathrm{sur}}(J_{X',\mathcal M}(\mathcal R))<\kappa^{-1}}
    \leq d^{-30}\enspace.
    \label{eq:tv_average_rotation_failure}
\end{equation}
Thus, applying \cref{lem:tv-gate-correction} with initial matrices
$X',Y'$ gives a coupling of the rotation lists for the remaining
$T_2$ steps, each with distribution $\nu_m$, such that
\begin{equation}
    \prob{}{\left.X_{T_3}\ne Y_{T_3}\,\right\vert X',Y',\mathcal M}
    \leq\eta_{\mathcal M}+C\kappa^4d^9T_2^3\norm{X'-Y'}_F\enspace.
    \label{eq:tv_fixed_matching_coupling}
\end{equation}
Sample the same matching sequence for both walks and then use this
conditional coupling of their rotation lists.
Averaging \cref{eq:tv_fixed_matching_coupling} over $\mathcal M$, we obtain
\begin{align*}
    \prob{}{\left.X_{T_3}\ne Y_{T_3}\,\right\vert X',Y'}
    &\leq\expect{\mathcal M}{\eta_{\mathcal M}}
        +C\kappa^4d^9T_2^3\norm{X'-Y'}_F \leq d^{-30}+d^{-22} \leq d^{-21}\enspace,
\end{align*}
where the second inequality uses
\cref{eq:tv_average_rotation_failure,eq:tv_correction_error_bound}.
If $\norm{X'-Y'}_F>d^{-116}$, use any coupling for the remaining steps.
Together with \cref{eq:tv_close_states_failure}, this yields
\[
    \prob{}{X_{T_3}\ne Y_{T_3}}
    \leq\prob{}{\norm{X'-Y'}_F>d^{-116}}+d^{-21}
    \leq d^{-24}+d^{-21} \leq d^{-20}\enspace,
\]
proving \cref{eq:tv_block_coupling}.

\paragraph{Iterating the coupling.}
Now start the first walk from an arbitrary $X\in\stframe{d}{k}$ and
sample the initial state $Y$ of the second walk from $\mu_{d,k}$.
Apply the coupling in \cref{eq:tv_block_coupling} in successive blocks
of $T_3$ steps; once the states agree, use identical updates thereafter.
Thus, for every integer $s\geq1$,
\[
    \prob{}{X_{(s+1)T_3}\ne Y_{(s+1)T_3}}
    \leq d^{-20}\cdot \prob{}{X_{sT_3}\ne Y_{sT_3}}\enspace.
\]
It follows by induction that
$\prob{}{X_{sT_3}\ne Y_{sT_3}}\leq d^{-20s}$.
By \cref{lem:coupling_lemma}, we therefore have
\begin{equation*}
    \sup_{X\in\stframe{d}{k}}
    \tv{\kackernel^{sT_3}(X,\cdot)-\mu_{d,k}}
    \leq d^{-20s}\enspace.
\end{equation*}
For $\varepsilon\in(0,1)$, take
$s=\lceil\log(1/\varepsilon)/(20\log d)\rceil$.
Therefore, we have
$t_{\mathrm{mix}}(\varepsilon)\leq sT_3$.
Since $T_3\leq2^{20}(k+\log d)\log d$, we have
\begin{align*}
    t_{\mathrm{mix}}(\varepsilon)
    &\leq sT_3 \leq2^{20}(k+\log d)
        \left(\log d+\frac{\log(1/\varepsilon)}{20}\right) \leq2^{20}(k+\log d)\log(d/\varepsilon)\enspace.\qedhere
\end{align*}
\end{proof}

\endgroup

\newpage
\bibliographystyle{alpha}
\bibliography{ref}

\newpage
\begingroup
\setstretch{1.1}

\newpage

\appendix
\hypersetup{bookmarksdepth=1}
\section*{Appendix}
\addcontentsline{toc}{section}{Appendix}
\addtocontents{toc}{\protect\setcounter{tocdepth}{0}}
\section{\texorpdfstring{Proof of \cref{clm:pair_observation}}{Proof of Claim \ref*{clm:pair_observation}}}
\label[appendix]{app:pair_observation}

\pairObservationClaim*

\begin{proof}
Since $\sigma_z^\dagger=-\sigma_z$ for $z\in\{1,2,3\}$,
we have that
$ (X_{(i,j)}^\dagger\sigma_zX_{(i,j)} )^\dagger = - (X_{(i,j)}^\dagger\sigma_zX_{(i,j)} ) $
which means that the diagonal entries of $X_{(i,j)}^\dagger\sigma_zX_{(i,j)}$ are purely imaginary.
Thus, for $z\in\{1,2,3\}$,
\begin{equation}
    \re\!\br{
        \tr{X_{(i,j)}^\dagger\sigma_zX_{(i,j)}}
    }
    =0\enspace. \label{eq:x_sigma_x_zero}
\end{equation}

Let $K=Y-X$ and $P=\frac12\cdot(KX^\dagger-XK^\dagger)$.
For every $i<j$, let $K_{(i,j)}=Y_{(i,j)}-X_{(i,j)}$.
By \cref{eq:x_sigma_x_zero}, we have
\begin{align*}
    p_{(i,j),z}
    ~&=~
    \re\!\br{
        \tr{X_{(i,j)}^\dagger\sigma_zK_{(i,j)}}
    }
    ~=~
    \frac12\cdot\re\!\br{
        \tr{\br{
            K_{(i,j)}X_{(i,j)}^\dagger
            -X_{(i,j)}K_{(i,j)}^\dagger
        }\sigma_z}
    } \\
    ~&=~
    \re\!\br{
        \tr{
            \begin{pmatrix}
                P_{ii}&P_{ij}\\
                -\overline{P_{ij}}&P_{jj}
            \end{pmatrix}
            \sigma_z
        }
    }\enspace,
\end{align*}
where we use $P = -P^\dagger$ in the last step.
Substituting the three
matrices in \cref{eq:basis_for_su2} gives
\begin{align*}
    p_{(i,j),1}
    ~=~
    -2\cdot\re\!\br{P_{ ij}} \enspace,\qquad
    p_{(i,j),2}
    ~=~
    -2\cdot\im\!\br{P_{ij}}\enspace,\qquad
    p_{(i,j),3}
    ~=~
    \mathrm{i}\cdot\br{P_{ii}-P_{jj}}\enspace.
\end{align*}
Thus,
$
    \sum_{z=1}^{3}p_{(i,j),z}^2
    =
    \abs{P_{ii}-P_{jj}}^2+4\abs{P_{ij}}^2
$.
Consequently,
\begin{align}
    \sum_{i<j}\sum_{z=1}^{3}p_{(i,j),z}^2
    ~=~
    d\cdot\sum_{i=1}^{d}
    \abs{P_{ii}-\frac{\tr{P}}{d}}^2
    +4\cdot\sum_{i<j}\abs{P_{ij}}^2
    ~\geq~
    2\cdot\norm{P-\frac{\tr{P}}{d}\cdot I_d}_F^2\enspace,
    \label{eq:pair_observation_centered_P}
\end{align}
where we used
$\sum_{i<j}\abs{P_{ii}-P_{jj}}^2
=d\sum_i\abs{P_{ii}-\tr{P}/d}^2$ and $d\geq2$.

Now we write
$X^\dagger K = S + A$
with $ S = \frac{1}{2}( X^\dagger K + K^\dagger X ) = -\frac12\cdot K^\dagger K $ 
and $ A = \frac{1}{2}( X^\dagger K - K^\dagger X )  $.
Moreover, decomposing
$K=XX^\dagger K+(I_d-XX^\dagger)K$ yields
\begin{align*}
    P
    ~=~
    XAX^\dagger
    +\frac12\cdot\br{
        (I_d-XX^\dagger)KX^\dagger
        -XK^\dagger(I_d-XX^\dagger)
    }\enspace.
\end{align*}
Since $X^\dagger X = I_k$, we have
$X^\dagger(I_d - XX^\dagger) = (I_d - XX^\dagger)X = 0 $.
Therefore, we have that
$XAX^\dagger$, $(I_d-XX^\dagger)KX^\dagger$ and $XK^\dagger(I_d-XX^\dagger)$
are pairwise orthogonal
with respect to the Frobenius inner product.
Hence, we have
\begin{align*}
    \norm{P}_F^2
    ~=~
    \norm{A}_F^2
    +\frac12\cdot\norm{(I_d-XX^\dagger)K}_F^2 \enspace,
    \qquad
    \tr{P}=\tr{A}\enspace.
\end{align*}
Note that
\begin{align*}
    \norm{P-\frac{\tr{P}}{d}\cdot I_d}_F^2
    ~&=~
    \tr{ P^\dagger P - \frac{\tr{P}}{d}\cdot P^\dagger - \frac{\overline{\tr{P}}}{d}\cdot P + \frac{\abs{\tr{P}}^2}{d^2}\cdot I_d }
    ~=~
    \norm{P}_F^2 - \frac{\abs{\tr{P}}^2}{d} \enspace.
\end{align*}
Using $\abs{\tr{A}}^2\leq k\norm{A}_F^2$ and $k\leq d/2$, we obtain
\begin{align}
    \norm{P-\frac{\tr{P}}{d}\cdot I_d}_F^2 
    ~&=~
    \norm{A}_F^2
    +\frac12\cdot\norm{(I_d-XX^\dagger)K}_F^2
    - \frac{\abs{\tr{A}}^2}{d} \nonumber\\
    ~&\geq~
    \br{1-\frac{k}{d}}\cdot\norm{A}_F^2
    +\frac12\cdot\norm{(I_d-XX^\dagger)K}_F^2 \nonumber\\
    ~&\geq~
    \frac12\cdot\br{
        \norm{A}_F^2+\norm{(I_d-XX^\dagger)K}_F^2
    } \nonumber \\
    ~&=~
    \frac12\cdot\br{\norm{K}_F^2-\norm{S}_F^2}\enspace,
    \label{eq:pair_observation_lower_bound_on_P}
\end{align}
where the last equality holds because
$\re\!\br{ \tr{S^\dagger A} } = 0$ and therefore
\begin{align*}
    \norm{K}_F^2
    ~=~
    \norm{X^\dagger K}_F^2+\norm{(I_d-XX^\dagger)K}_F^2
    ~=~
    \norm{S}_F^2+\norm{A}_F^2+\norm{(I_d-XX^\dagger)K}_F^2\enspace.
\end{align*}
Since $S=-\frac12\cdot K^\dagger K$, we have
$ \norm{S}_F^2 \leq \frac14\cdot\norm{K}_F^4 $.
Combining
\cref{eq:pair_observation_centered_P,eq:pair_observation_lower_bound_on_P},
 we get
\begin{align*}
    \sum_{i<j}\sum_{z=1}^{3}p_{(i,j),z}^2
    ~\geq~
    \norm{K}_F^2-
    \frac14\cdot\norm{K}_F^4
    ~\geq~
    \frac34\cdot\norm{K}_F^2 
    ~=~
    \frac34\cdot\norm{X-Y}_F^2\enspace,
\end{align*}
where the second inequality uses $\norm{X-Y}_F=\norm{K}_F\leq1$.
This proves the claim.
\end{proof}

\section{\texorpdfstring{Proof of \cref{clm:su_pair_observation}}{Proof of Claim \ref*{clm:su_pair_observation}}}
\label[appendix]{app:su_pair_observation}

\suPairObservationClaim*

\begin{proof}
    Let $K=Y-X$, $P=\frac12\cdot(KX^\dagger-XK^\dagger)$
    and 
    $ \widehat{P} = P-\frac{\tr{P}}{d}\cdot I_d $.
    By \cref{eq:pair_observation_centered_P}, we have
    $
        \sum_{i<j}\sum_{z=1}^3p_{(i,j),z}^2
        \geq
        2\Vert \widehat{P} \Vert_F^2
    $. 
Thus, it suffices to show that
$\Vert\widehat{P}\Vert_F^2\geq\frac34\cdot\norm{X-Y}_F^2$.
Let $\e^{\i\theta_1},\dots,\e^{\i\theta_d}$ be the eigenvalues of $YX^\dagger$,
where $\theta_a\in(-\pi,\pi]$.
Since $P=\frac12\cdot(YX^\dagger-XY^\dagger)$, the eigenvalues of $P$ are
$\i\sin\theta_1,\dots,\i\sin\theta_d$.
Therefore, we have
\begin{align*}
    \Vert\widehat{P}\Vert_F^2
    ~&=~
    \sum_{a=1}^d\br{
        \sin\theta_a-\frac1d\sum_{b=1}^d\sin\theta_b
    }^2
    ~=~
    \frac1d\sum_{a<b}(\sin\theta_a-\sin\theta_b)^2
    ~\geq~
    \frac34\sum_{a=1}^d\theta_a^2
    ~\geq~
    \frac34\cdot\norm{X-Y}_F^2\enspace,
\end{align*}
where the first inequality is proved below, and the last inequality uses
the fact that
$\norm{X-Y}_F^2=4\sum_a\sin^2(\theta_a/2)\leq\sum_a\theta_a^2$.
To prove the first inequality, we use
$\abs{\theta}\leq\frac\pi2\abs{1-\e^{\i\theta}}$
for $\abs{\theta}\leq\pi$. By
\cref{lem:frobenius_riemannian_distance_comparison},
\[
    \abs{\theta_a}
    ~\leq~\frac\pi2\cdot\norm{X-Y}_F
    ~\leq~\frac\pi2\cdot D(X,Y)
    ~\leq~\frac{\pi}{32(d-1)}\enspace.
\]
Since $\det YX^\dagger=1$ and
$\abs{\sum_a\theta_a}\leq d\pi/(32(d-1))<2\pi$,
we have $\sum_a\theta_a=0$.
Moreover, $\abs{\theta_a}\leq\pi/32<\pi/6$ for every $a$,
so the mean value theorem gives
\[
    \abs{\sin\theta_a-\sin\theta_b}
    ~\geq~\frac{\sqrt3}{2}\cdot\abs{\theta_a-\theta_b}\enspace.
\]
Squaring and summing over $a<b$, and using
$\sum_{a<b}(\theta_a-\theta_b)^2=d\sum_a\theta_a^2$,
proves the first inequality and completes the proof.
\end{proof}

\section{\texorpdfstring{Proof of \cref{clm:tv_correction_invertibility,clm:tv_correction_flow_bounds}}{Proof of Claims 5.4 and 5.5}}
\label[appendix]{app:tv_correction_claims}

We use the notation and hypotheses of \cref{lem:tv-gate-correction},
with $C=2^{20}$. In particular,
$\|X-Y\|_F\le(C\kappa^4d^9T^3)^{-1}\le1/4$.
We first establish some bounds needed in the proofs below.
We then prove \cref{clm:tv_correction_invertibility,clm:tv_correction_flow_bounds}
in \cref{app:tv_correction_invertibility_proof,app:tv_correction_flow_bounds_proof}, respectively.

\subsection{Bounds Used in the Proof}
\label{subsec:tv_correction_derivatives}

For fixed $\mathcal R$, we regard $J_{X_q,\mathcal M}(\mathcal R)$
as a linear map from $\sualg(2)^m$ into the fixed ambient space
$\C^{d\times k}$ and differentiate with respect to $q$ in this space.
We define
\[
  \|\partial_qJ_{X_q,\mathcal M}(\mathcal R)\|_{\mathrm{op}}
  \coloneqq\sup_{\|\mathcal B\|_F=1}
    \left\|\frac{\mathrm d}{\mathrm d q}
      \bigl(J_{X_q,\mathcal M}(\mathcal R)[\mathcal B]\bigr)
    \right\|_F\enspace,
\]
where $\mathcal B\in\sualg(2)^m$ is held fixed when differentiating.

\begin{claim}
\label{clm:tv_basic_bounds}
For all $q\in[0,1]$ and $\mathcal R\in\su(2)^m$,
$\|J_{X_q,\mathcal M}(\mathcal R)\|_{\mathrm{op}}\le\sqrt m$,
$\|\partial_qJ_{X_q,\mathcal M}(\mathcal R)\|_{\mathrm{op}}
\le2\sqrt m\,\|X-Y\|_F$, and
$\|H_{\mathcal R}(q)\|_F\le2\|X-Y\|_F$.
\end{claim}

\begin{proof}
Write $E=Y-X$.
Since $Z_q^\dagger Z_q=I_k-q(1-q)E^\dagger E\succeq(63/64)I_k$ is invertible,
we have
\[
  \dot X_q
  =E(Z_q^\dagger Z_q)^{-1/2}
   +\frac{1-2q}{2}Z_q(Z_q^\dagger Z_q)^{-3/2}E^\dagger E\enspace.
\]
Since $Z_q^\dagger Z_q\succeq(63/64)I_k$, we have
\[
  \|(Z_q^\dagger Z_q)^{-1/2}\|_{\mathrm{op}}
  \le\left(\frac{64}{63}\right)^{1/2},
  \qquad
  \|(Z_q^\dagger Z_q)^{-3/2}\|_{\mathrm{op}}
  \le\left(\frac{64}{63}\right)^{3/2}.
\]
Moreover, $\|Z_q\|_{\mathrm{op}}\le1$, $|1-2q|\le1$, and
\(
  \|E^\dagger E\|_F
  \le\|E\|_{\mathrm{op}}\|E\|_F
  \le\|E\|_F^2.
\)
Therefore, using $\|E\|_F=\|X-Y\|_F\le1/4$, we obtain
\[
  \|\dot X_q\|_F
  \le
  \left(\frac{64}{63}\right)^{1/2}\|E\|_F
  +\frac12\left(\frac{64}{63}\right)^{3/2}\|E\|_F^2 \le
  \left[
    \left(\frac{64}{63}\right)^{1/2}
    +\frac18\left(\frac{64}{63}\right)^{3/2}
  \right]\|E\|_F
  \le2\|X-Y\|_F \enspace.
\]
For a rotation list $\mathcal R=(R_{t,e})_{t,e}$, write
$R_1,\dots,R_m$ for the rotations in their order of application.
For $\mathcal B=(B_i)_{i=1}^m\in\sualg(2)^m$, the product rule gives
\begin{equation}
\label{eq:tv_jacobian_product}
  J_{X_q,\mathcal M}(\mathcal R)[\mathcal B]
  =\sum_{i=1}^m R_m\cdots R_iB_iR_{i-1}\cdots R_1X_q \enspace.
\end{equation}
Since the rotations are unitary and $\|X_q\|_{\mathrm{op}}=1$,
\[
  \|J_{X_q,\mathcal M}(\mathcal R)[\mathcal B]\|_F
  \le\sum_{i=1}^m\|B_i\|_F \le\sqrt m
    \left(\sum_{i=1}^m\|B_i\|_F^2\right)^{1/2}
   =\sqrt m\,\|\mathcal B\|_F\enspace,
\]
where the second inequality is Cauchy--Schwarz.
Thus, $\|J_{X_q,\mathcal M}(\mathcal R)\|_{\mathrm{op}}
\le\sqrt m$.
Since
\[
  (\partial_qJ_{X_q,\mathcal M}(\mathcal R))[\mathcal B]
  =\sum_{i=1}^m
    R_m\cdots R_iB_iR_{i-1}\cdots R_1\dot X_q\enspace,
\]
we have
\[
\begin{aligned}
  \|(\partial_qJ_{X_q,\mathcal M}(\mathcal R))[\mathcal B]\|_F
  &\le\|\dot X_q\|_F\sum_{i=1}^m\|B_i\|_F \le2\sqrt m\,\|X-Y\|_F\|\mathcal B\|_F.
\end{aligned}
\]
Therefore,
\(
  \|\partial_qJ_{X_q,\mathcal M}(\mathcal R)\|_{\mathrm{op}}
  \le2\sqrt m\,\|X-Y\|_F.
\)
Finally, since $H_{\mathcal R}(q)=R_m\cdots R_1\dot X_q$ and the
rotations are unitary,
\[
  \|H_{\mathcal R}(q)\|_F=\|\dot X_q\|_F
  \le2\|X-Y\|_F \enspace. \qedhere
\]
\end{proof}

\begin{claim}
\label{clm:tv_q_derivative_bounds}
For all $q\in[0,1]$ and $\mathcal R\in\su(2)^m$,
\(
  \|A_{q,\mathcal R}-A_{0,\mathcal R}\|_{\mathrm{op}}
  \le(4m+4)\|X-Y\|_F.
\)
\end{claim}

\begin{proof}
Since
\(
  A_{q,\mathcal R}-A_{0,\mathcal R}
  =\int_0^q\partial_sA_{s,\mathcal R}\,\mathrm ds
\),
we have
\[
	\norm{ A_{q,\mathcal R}-A_{0,\mathcal R} }_{\mathrm{op}} \leq \int_0^q \norm{ \partial_s A_{s,\mathcal R} }_{\mathrm{op}} \, \mathrm ds \enspace.
\]
Therefore, it is sufficient to prove that for all $q\in[0,1]$ and $\mathcal R\in\su(2)^m$,
\[
  \|\partial_qA_{q,\mathcal R}\|_{\mathrm{op}}
  \le(4m+4)\|X-Y\|_F \enspace.
\]

\paragraph{Bound on $\left\|\partial_q\Pi_{\Phi_{X_q,\mathcal M}(\mathcal R)}\right\|_{\mathrm{op}}$}
For any matrix $V$, we have
\begin{equation}
\label{eq:tv_projection_formula}
  \Pi_{\Phi_{X_q,\mathcal M}(\mathcal R)}(V)
  =V-\frac12\Phi_{X_q,\mathcal M}(\mathcal R)
  \bigl(\Phi_{X_q,\mathcal M}(\mathcal R)^\dagger V
        +V^\dagger\Phi_{X_q,\mathcal M}(\mathcal R)\bigr) \enspace.
\end{equation}
Since $\partial_q\Phi_{X_q,\mathcal M}(\mathcal R)=H_{\mathcal R}(q)$,
we have
\[
\begin{aligned}
  \left(\partial_q
    \Pi_{\Phi_{X_q,\mathcal M}(\mathcal R)}\right)(V)
  ={}&-\frac12H_{\mathcal R}(q)
    \bigl(\Phi_{X_q,\mathcal M}(\mathcal R)^\dagger V
          +V^\dagger\Phi_{X_q,\mathcal M}(\mathcal R)\bigr)\\
  &\qquad-\frac12\Phi_{X_q,\mathcal M}(\mathcal R)
    \bigl(H_{\mathcal R}(q)^\dagger V
          +V^\dagger H_{\mathcal R}(q)\bigr) \enspace.
\end{aligned}
\]
Since $\|\Phi_{X_q,\mathcal M}(\mathcal R)\|_{\mathrm{op}}=1$,
the four terms on the right-hand side give
\[
  \left\|\partial_q
    \Pi_{\Phi_{X_q,\mathcal M}(\mathcal R)}\right\|_{\mathrm{op}}
  \le2\|H_{\mathcal R}(q)\|_F\le4\|X-Y\|_F \enspace,
\]
where the last inequality follows from \cref{clm:tv_basic_bounds}.

Now we write $J_q=J_{X_q,\mathcal M}(\mathcal R)$ and
$\Pi_q=\Pi_{\Phi_{X_q,\mathcal M}(\mathcal R)}$,
and we have
\[
  \partial_qA_{q,\mathcal R}
  =(\partial_qJ_q)J_q^\top
   +J_q(\partial_qJ_q)^\top
   -\partial_q\Pi_q \enspace.
\]
Consequently, using \cref{clm:tv_basic_bounds},
\[
\begin{aligned}
  \|\partial_qA_{q,\mathcal R}\|_{\mathrm{op}}
  &\le
    2\|J_q\|_{\mathrm{op}}\|\partial_qJ_q\|_{\mathrm{op}}
    +\|\partial_q\Pi_q\|_{\mathrm{op}}\\
  &\le
    2\sqrt m\bigl(2\sqrt m\,\|X-Y\|_F\bigr)
    +4\|X-Y\|_F\\
  &=(4m+4)\|X-Y\|_F \enspace.
\end{aligned}
\]
This proves the required derivative bound and hence the claim.
\end{proof}

For a differentiable map $g$ on $\su(2)^m$ and
$\mathcal R=(R_{t,e})_{t\in[T],\,e\in M_t}\in\su(2)^m$, define
\[
  \bigl(\mathrm d_{\mathcal R}g(\mathcal R)\bigr)[\mathcal C]
  \coloneqq\left.\frac{\mathrm d}{\mathrm ds}
    g\bigl((R_{t,e}\exp(sC_{t,e}))_{t\in[T],\,e\in M_t}\bigr)
    \right|_{s=0}\enspace,
\]
where $\mathcal C=(C_{t,e})\in\sualg(2)^m$.
The three derivative norms used below are
\[
\begin{aligned}
  \|\mathrm d_{\mathcal R}J_{X_q,\mathcal M}(\mathcal R)\|_{\mathrm{op}}
  &\coloneqq\sup_{\|\mathcal C\|_F=1}
    \left\|\bigl(\mathrm d_{\mathcal R}J_{X_q,\mathcal M}(\mathcal R)\bigr)
      [\mathcal C]\right\|_{\mathrm{op}}\enspace,\\
  \|\mathrm d_{\mathcal R}H_{\mathcal R}(q)\|_{\mathrm{op}}
  &\coloneqq\sup_{\|\mathcal C\|_F=1}
    \|(\mathrm d_{\mathcal R}H_{\mathcal R}(q))[\mathcal C]\|_F\enspace,\\
  \|\mathrm d_{\mathcal R}A_{q,\mathcal R}\|_{\mathrm{op}}
  &\coloneqq\sup_{\|\mathcal C\|_F=1}
    \left\|\bigl(\mathrm d_{\mathcal R}A_{q,\mathcal R}\bigr)
      [\mathcal C]\right\|_{\mathrm{op}}\enspace,
\end{aligned}
\]
where $\mathcal C\in\sualg(2)^m$.
The directional derivatives of $J$ and $A$ are linear operators,
so we use their operator norms induced by the Frobenius norms.
The directional derivative of $H$ is a matrix, so we use its
Frobenius norm.

\begin{claim}
\label{clm:tv_rotation_derivative_bounds}
For all $q\in[0,1]$ and $\mathcal R\in\su(2)^m$,
$\|\mathrm d_{\mathcal R}J_{X_q,\mathcal M}(\mathcal R)\|_{\mathrm{op}}
\le m$,
$\|\mathrm d_{\mathcal R}H_{\mathcal R}(q)\|_{\mathrm{op}}
\le2\sqrt m \|X-Y\|_F$, and
$\|\mathrm d_{\mathcal R}A_{q,\mathcal R}\|_{\mathrm{op}}
\le4m^{3/2}$.
\end{claim}

\begin{proof}
For unit directions $\mathcal B=(B_i)$ and $\mathcal C=(C_i)$,
differentiating the product formula \cref{eq:tv_jacobian_product}
in direction $\mathcal C$ gives at most
$m^2$ terms, whose norms sum to at most
\[
\begin{aligned}
  \sum_{i,j=1}^m\|B_i\|_F\|C_j\|_F
  &=\left(\sum_{i=1}^m\|B_i\|_F\right)
    \left(\sum_{j=1}^m\|C_j\|_F\right)
  \le m\|\mathcal B\|_F\|\mathcal C\|_F \enspace,
\end{aligned}
\]
where the inequality uses Cauchy--Schwarz inequality.
Thus, $\|\mathrm d_{\mathcal R}J_{X_q,\mathcal M}(\mathcal R)\|_{\mathrm{op}}\le m$.
Similarly,
\[
  (\mathrm d_{\mathcal R}H_{\mathcal R}(q))[\mathcal C]
  =\sum_{i=1}^m
    R_m\cdots R_iC_iR_{i-1}\cdots R_1\dot X_q \enspace.
\]
Since the rotations are unitary, the triangle inequality and
Cauchy--Schwarz inequality give
\[
\begin{aligned}
  \|\mathrm d_{\mathcal R}H_{\mathcal R}(q)[\mathcal C]\|_F
  &\le\|\dot X_q\|_F\sum_{i=1}^m\|C_i\|_F
  \le2\sqrt m\,\|X-Y\|_F\|\mathcal C\|_F \enspace,
\end{aligned}
\]
where the last inequality uses $\|\dot X_q\|_F\le2\|X-Y\|_F$
from the proof of \cref{clm:tv_basic_bounds}.
Taking the supremum over $\|\mathcal C\|_F=1$ yields
$\|\mathrm d_{\mathcal R}H_{\mathcal R}(q)\|_{\mathrm{op}}
\le2\sqrt m\,\|X-Y\|_F$.

Write $\Phi=\Phi_{X_q,\mathcal M}(\mathcal R)$ and
$J=J_{X_q,\mathcal M}(\mathcal R)$, with $q$ fixed.
Since $(\mathrm d_{\mathcal R}\Phi)[\mathcal C]=J[\mathcal C]$,
differentiating the projection formula \cref{eq:tv_projection_formula}
gives, for every $V\in\C^{d\times k}$,
\[
\begin{aligned}
  \bigl(\mathrm d_{\mathcal R}\Pi_{\Phi}[\mathcal C]\bigr)(V)
  ={}&-\frac12J[\mathcal C](\Phi^\dagger V+V^\dagger \Phi) -\frac12\Phi\bigl(J[\mathcal C]^\dagger V
                         +V^\dagger J[\mathcal C]\bigr)\enspace.
\end{aligned}
\]
Using $\|\Phi\|_{\mathrm{op}}=1$ and \cref{clm:tv_basic_bounds},
we obtain
\[
\begin{aligned}
  \|\bigl(\mathrm d_{\mathcal R}\Pi_{\Phi}[\mathcal C]\bigr)(V)\|_F
  &\le2\|J[\mathcal C]\|_F\|V\|_F \le2\|J\|_{\mathrm{op}}\|\mathcal C\|_F\|V\|_F
   \le2\sqrt m\,\|\mathcal C\|_F\|V\|_F\enspace.
\end{aligned}
\]
Taking the supremum over $\|\mathcal C\|_F=\|V\|_F=1$ yields
$\|\mathrm d_{\mathcal R}\Pi_{\Phi}\|_{\mathrm{op}}\le2\sqrt m$.
Finally, differentiating $A_{q,\mathcal R}=JJ^\top+I-\Pi_{\Phi}$ gives
\[
  \mathrm d_{\mathcal R}A_{q,\mathcal R}[\mathcal C]
  =\bigl(\mathrm d_{\mathcal R}J[\mathcal C]\bigr)J^\top
   +J\bigl(\mathrm d_{\mathcal R}J[\mathcal C]\bigr)^\top
   -\mathrm d_{\mathcal R}\Pi_{\Phi}[\mathcal C]\enspace.
\]
Thus, by the bounds on $J$, $\mathrm d_{\mathcal R}J$, and
$\mathrm d_{\mathcal R}\Pi_{\Phi}$,
\[
\begin{aligned}
  \|\mathrm d_{\mathcal R}A_{q,\mathcal R}[\mathcal C]\|_{\mathrm{op}}
  &\le2\|J\|_{\mathrm{op}}
       \|\mathrm d_{\mathcal R}J[\mathcal C]\|_{\mathrm{op}}
       +\|\mathrm d_{\mathcal R}\Pi_{\Phi}[\mathcal C]\|_{\mathrm{op}}\\
  &\le\bigl(2\sqrt m\cdot m+2\sqrt m\bigr)\|\mathcal C\|_F\\
  &\le4m^{3/2}\|\mathcal C\|_F\enspace.
\end{aligned}
\]
Taking the supremum over $\|\mathcal C\|_F=1$ proves
\(
  \|\mathrm d_{\mathcal R}A_{q,\mathcal R}\|_{\mathrm{op}}
  \le4m^{3/2}\enspace.
\)
\end{proof}

For $f$ and $\chi$ defined in the proof of \cref{lem:tv-gate-correction}
and each fixed $\mathcal R\in\su(2)^m$, we define:
\[
\begin{aligned}
  \|\mathrm d_{\mathcal R}f(\mathcal R)\|
  &\coloneqq\sup_{\|\mathcal C\|_F=1}
    \left|\bigl(\mathrm d_{\mathcal R}f(\mathcal R)\bigr)
      [\mathcal C]\right| \enspace,\\
  \|\mathrm d_{\mathcal R}(\chi\circ f)(\mathcal R)\|
  &\coloneqq\sup_{\|\mathcal C\|_F=1}
    \left|\bigl(\mathrm d_{\mathcal R}(\chi\circ f)(\mathcal R)\bigr)
      [\mathcal C]\right|\enspace.
\end{aligned}
\]

\begin{claim}
\label{clm:tv_cutoff_derivative_bounds}
$\|\mathrm d_{\mathcal R}f(\mathcal R)\|\le8192k^3d^3\kappa^4m^{3/2}$ and
$\|\mathrm d_{\mathcal R}(\chi\circ f)(\mathcal R)\|
\le4096k^2d^2\kappa^2m^{3/2}$
hold for all $\mathcal R\in\su(2)^m$.
\end{claim}

\begin{proof}
Differentiating the inverse in the definition of $f$ gives, for
$\mathcal C\in\sualg(2)^m$,
\[
  (\mathrm d_{\mathcal R}f(\mathcal R))[\mathcal C]
  =-\operatorname{tr}\Bigl(
    \bigl(A_{0,\mathcal R}+(32kd\kappa^2)^{-1}I\bigr)^{-1}
    \bigl(\mathrm d_{\mathcal R}A_{0,\mathcal R}[\mathcal C]\bigr)
    \cdot
    \bigl(A_{0,\mathcal R}+(32kd\kappa^2)^{-1}I\bigr)^{-1}\Bigr) \enspace.
\]
Since $A_{0,\mathcal R}\succeq0$,
$\bigl(A_{0,\mathcal R}+(32kd\kappa^2)^{-1}I\bigr)^{-1}$ has
operator norm at most $32kd\kappa^2$.
Taking the trace and using \cref{clm:tv_rotation_derivative_bounds} yields
\[
  \|\mathrm d_{\mathcal R}f(\mathcal R)\|
  \le(2dk)(32kd\kappa^2)^2(4m^{3/2})
  =8192k^3d^3\kappa^4m^{3/2} \enspace.
\]
Since $\abs{\chi'(x)}\leq 1/(2kd\kappa^2)$, we have
\[
  \|\mathrm d_{\mathcal R}(\chi\circ f)(\mathcal R)\|
  \le\frac{\|\mathrm d_{\mathcal R}f(\mathcal R)\|}{2kd\kappa^2}
  \le4096k^2d^2\kappa^2m^{3/2} \enspace.\qedhere
\]
\end{proof}

\subsection{\texorpdfstring{Proof of \cref{clm:tv_correction_invertibility}}{Invertibility of the correction operator}}
\label[appendix]{app:tv_correction_invertibility_proof}

\correctionInvertibilityClaim*

\begin{proof}
The definition of $A_{q,\mathcal R}$ gives the block decomposition
\[
  A_{q,\mathcal R}
  =\begin{pmatrix}
    J_{X_q,\mathcal M}(\mathcal R)
      J_{X_q,\mathcal M}(\mathcal R)^\top&0\\
    0&I
  \end{pmatrix} \enspace,
\]
where the upper-left block acts on the tangent space at $\Phi_{X_q,\mathcal M}(\mathcal R)$.
Fix $\mathcal R\in\mathfrak R_{\mathcal M}$.
For any tangent vector $H$ at $\Phi_{X,\mathcal M}(\mathcal R)$,
since $\sigma_{\mathrm{sur}}(J_{X,\mathcal M}(\mathcal R))\ge\kappa^{-1}$,
we have
\[
  \langle H,
    J_{X,\mathcal M}(\mathcal R)
    J_{X,\mathcal M}(\mathcal R)^\top H\rangle_F
  =\|J_{X,\mathcal M}(\mathcal R)^\top H\|_F^2
  \ge\kappa^{-2}\|H\|_F^2 \enspace.
\]
Hence $A_{0,\mathcal R}\succeq\kappa^{-2}I$ on the entire space, and
\[
  \bigl(A_{0,\mathcal R}+(32kd\kappa^2)^{-1}I\bigr)^{-1}
  \preceq
  \frac{1}{\kappa^{-2}+(32kd\kappa^2)^{-1}}I
  \prec\kappa^2I \enspace.
\]
Taking the trace proves $f(\mathcal R)<2kd\kappa^2$ for every $\mathcal R\in \mathfrak R_{\mathcal M}$.

Now fix $\mathcal R\in\Omega$.
Let $\lambda$ be the smallest eigenvalue of $A_{0,\mathcal R}$
and we have
\[
  \frac{1}{\lambda+(32kd\kappa^2)^{-1}}
  \le f(\mathcal R)<12kd\kappa^2 \enspace.
\]
Rearranging this, we obtain
\[
  \lambda>
  \frac{1}{12kd\kappa^2}-\frac{1}{32kd\kappa^2}
  =\frac{5}{96kd\kappa^2}
  >\frac{1}{20kd\kappa^2}.
\]
For any matrix $V$, we have
\begin{align*}
  \langle V,A_{q,\mathcal R}V\rangle_F
  &\ge\langle V,A_{0,\mathcal R}V\rangle_F
     -\|A_{q,\mathcal R}-A_{0,\mathcal R}\|_{\mathrm{op}}\|V\|_F^2 \ge\left(\frac{1}{20kd\kappa^2}
             -\frac{1}{40kd\kappa^2}\right)\|V\|_F^2
   =\frac{\|V\|_F^2}{40kd\kappa^2} \,,
\end{align*}
where the second inequality uses the lower bound on $\lambda$ and
the fact that for every $q\in[0,1]$,
\[
  \|A_{q,\mathcal R}-A_{0,\mathcal R}\|_{\mathrm{op}}
  \le(4m+4)\|X-Y\|_F
  \le\frac{4(m+1)}{C\kappa^4d^9T^3}
  \le\frac{1}{40kd\kappa^2}\enspace,
\]
where we use \cref{clm:tv_q_derivative_bounds} in the first inequality.
Hence, the smallest eigenvalue of $A_{q,\mathcal R}$ is at least
$(40kd\kappa^2)^{-1}$. Therefore, $A_{q,\mathcal R}$ is invertible with
$\|A_{q,\mathcal R}^{-1}\|_{\mathrm{op}}\le40kd\kappa^2$.
\end{proof}

\subsection{\texorpdfstring{Proof of \cref{clm:tv_correction_flow_bounds}}{Properties of the correction flow}}
\label[appendix]{app:tv_correction_flow_bounds_proof}

We first give bounds on $Q_q(\mathcal R)$ and its derivative.
For fixed $q$ and $\mathcal R\in\Omega$, we define
\[
  \|\mathrm d_{\mathcal R}Q_q(\mathcal R)\|_{\mathrm{op}}
  \coloneqq\sup_{\|\mathcal C\|_F=1}
    \left\|\bigl(\mathrm d_{\mathcal R}Q_q(\mathcal R)\bigr)
      [\mathcal C]\right\|_{\mathrm{op}}\enspace,
\]
where $\mathcal C\in\sualg(2)^m$.

\begin{claim}
\label{clm:tv_right_inverse_derivative_bounds}
For all $q\in[0,1]$ and $\mathcal R\in\Omega$,
$J_{X_q,\mathcal M}(\mathcal R)Q_q(\mathcal R)
=\Pi_{\Phi_{X_q,\mathcal M}(\mathcal R)}$,
$\|Q_q(\mathcal R)\|_{\mathrm{op}}\le40kd\kappa^2\sqrt m$ and
$\|\mathrm d_{\mathcal R}Q_q(\mathcal R)\|_{\mathrm{op}}
\le8192k^2d^2\kappa^4m^2$.
\end{claim}

\begin{proof}
Fix $q\in[0,1]$ and $\mathcal R\in\Omega$, and write
$J=J_{X_q,\mathcal M}(\mathcal R)$ and
$\Phi=\Phi_{X_q,\mathcal M}(\mathcal R)$.
For $V\in\mathsf T_{\Phi}\stframe{d}{k}$,
\(
  JQ_q(\mathcal R)V
  =JJ^\top A_{q,\mathcal R}^{-1}V =V.
\)
For $V\in(\mathsf T_{\Phi}\stframe{d}{k})^\perp$,
$A_{q,\mathcal R}^{-1}V=V$ and $J^\top V=0$, so
$JQ_q(\mathcal R)V=JJ^\top V=0$.
Therefore, $JQ_q(\mathcal R)=\Pi_{\Phi}$.
By \cref{clm:tv_basic_bounds,clm:tv_correction_invertibility},
\[
  \|Q_q(\mathcal R)\|_{\mathrm{op}}
  \le\|J_{X_q,\mathcal M}(\mathcal R)^\top\|_{\mathrm{op}}
     \|A_{q,\mathcal R}^{-1}\|_{\mathrm{op}}
  \le40kd\kappa^2\sqrt m \enspace.
\]
Differentiating
$A_{q,\mathcal R}A_{q,\mathcal R}^{-1}=I$ in a direction
$\mathcal C\in\sualg(2)^m$ gives, on $\Omega$,
\[
  \mathrm d_{\mathcal R}(A_{q,\mathcal R}^{-1})[\mathcal C]
  =-A_{q,\mathcal R}^{-1}
    \bigl(\mathrm d_{\mathcal R}A_{q,\mathcal R}[\mathcal C]\bigr)
    A_{q,\mathcal R}^{-1} \enspace.
\]
Since $Q_q(\mathcal R)=J_{X_q,\mathcal M}(\mathcal R)^\top A_{q,\mathcal R}^{-1}$,
the product rule gives
\[
  \mathrm d_{\mathcal R}Q_q(\mathcal R)[\mathcal C]
  =\bigl(\mathrm d_{\mathcal R}J_{X_q,\mathcal M}(\mathcal R)
       [\mathcal C]\bigr)^\top A_{q,\mathcal R}^{-1}
  -J_{X_q,\mathcal M}(\mathcal R)^\top A_{q,\mathcal R}^{-1}
    \bigl(\mathrm d_{\mathcal R}A_{q,\mathcal R}[\mathcal C]\bigr)
    A_{q,\mathcal R}^{-1} \enspace.
\]
Thus, by \cref{clm:tv_basic_bounds,clm:tv_rotation_derivative_bounds,clm:tv_correction_invertibility},
we have
\[
  \|\mathrm d_{\mathcal R}Q_q(\mathcal R)\|_{\mathrm{op}}
  \le40kd\kappa^2m+6400k^2d^2\kappa^4m^2
  \le8192k^2d^2\kappa^4m^2 \enspace. \qedhere
\]
\end{proof}

We next bound $\mathcal B^{q,\mathcal R}$ and its derivative with respect to $\mathcal R$.
For fixed $q$, we define
\[
  \|\mathrm d_{\mathcal R}\mathcal B^{q,\mathcal R}\|_{\mathrm{op}}
  \coloneqq\sup_{\|\mathcal C\|_F=1}
    \left\|\bigl(\mathrm d_{\mathcal R}\mathcal B^{q,\mathcal R}\bigr)
      [\mathcal C]\right\|_F\enspace,
\]
where $\mathcal C\in\sualg(2)^m$.

\begin{claim}
\label{clm:tv_vector_field_coefficient_bounds}
The map $\mathcal R\mapsto\mathcal B^{q,\mathcal R}$ is continuously
differentiable, and both $\mathcal B^{q,\mathcal R}$ and
$\mathrm d_{\mathcal R}\mathcal B^{q,\mathcal R}$ are continuous in $q$.
For all $q\in[0,1]$ and $\mathcal R\in\su(2)^m$,
$\|\mathcal B^{q,\mathcal R}\|_F
\le80kd\kappa^2\sqrt m\,\|X-Y\|_F$ and
$\|\mathrm d_{\mathcal R}\mathcal B^{q,\mathcal R}\|_{\mathrm{op}}
\le2^{19}k^3d^3\kappa^4m^2\|X-Y\|_F$.
\end{claim}

\begin{proof}
Since $\chi\circ f$ is supported on
$\{\mathcal R:f(\mathcal R)\le8kd\kappa^2\}\subset\Omega$,
it is identically zero outside $\Omega$.
On $\Omega$, the maps $Q_q(\mathcal R)$ and $H_{\mathcal R}(q)$
are continuously differentiable in $\mathcal R$, and
$Q_q(\mathcal R)$, $H_{\mathcal R}(q)$,
$\mathrm d_{\mathcal R}Q_q(\mathcal R)$, and
$\mathrm d_{\mathcal R}H_{\mathcal R}(q)$ are continuous in $q$.
Therefore, the map $\mathcal R\mapsto\mathcal B^{q,\mathcal R}$ is
continuously differentiable on $\su(2)^m$, and both
$\mathcal B^{q,\mathcal R}$ and
$\mathrm d_{\mathcal R}\mathcal B^{q,\mathcal R}$ are continuous in $q$.
Since $0\le\chi\le1$, \cref{clm:tv_basic_bounds,clm:tv_right_inverse_derivative_bounds} give
that for $\mathcal R\in\Omega$,
\[
  \|\mathcal B^{q,\mathcal R}\|_F
  \le\|Q_q(\mathcal R)\|_{\mathrm{op}}\|H_{\mathcal R}(q)\|_F
  \le80kd\kappa^2\sqrt m\,\|X-Y\|_F \enspace.
\]
Outside $\Omega$, this bound holds because $\mathcal B^{q,\mathcal R}=0$.
On $\Omega$, differentiating $\mathcal B^{q,\mathcal R}$ gives
\begin{align*}
  \mathrm d_{\mathcal R}\mathcal B^{q,\mathcal R}[\mathcal C]
  &=-\bigl(\mathrm d_{\mathcal R}(\chi\circ f)(\mathcal R)[\mathcal C]\bigr)
      Q_q(\mathcal R)H_{\mathcal R}(q)
  -\chi(f(\mathcal R))
      \bigl(\mathrm d_{\mathcal R}Q_q(\mathcal R)[\mathcal C]\bigr)
      H_{\mathcal R}(q)\\
  &\qquad\qquad-\chi(f(\mathcal R))Q_q(\mathcal R)
      \bigl(\mathrm d_{\mathcal R}H_{\mathcal R}(q)[\mathcal C]\bigr) \enspace.
\end{align*}
Applying \cref{clm:tv_basic_bounds,clm:tv_rotation_derivative_bounds,clm:tv_cutoff_derivative_bounds,clm:tv_right_inverse_derivative_bounds}
we have
\begin{align*}
  \|\mathrm d_{\mathcal R}\mathcal B^{q,\mathcal R}\|_{\mathrm{op}}
  &\leq
  (4096k^2d^2\kappa^2m^{3/2})(40kd\kappa^2\sqrt m)(2\|X-Y\|_F) +
  (8192k^2d^2\kappa^4m^2)(2\|X-Y\|_F) \\
  &\qquad\qquad\qquad+(40kd\kappa^2\sqrt m)(2\sqrt m\,\|X-Y\|_F) \\
  &\le2^{19}k^3d^3\kappa^4m^2\|X-Y\|_F \enspace.
\end{align*}
Outside $\Omega$, the derivative is zero, so the bound also holds.
\end{proof}

We now prove \cref{clm:tv_correction_flow_bounds}
and we restate it here for the reader's convenience.

\correctionFlowBoundsClaim*

\begin{proof}
Define
\[
  \mathcal G(q,\mathcal R)
  \coloneqq(R_{t,e}B_{t,e}^{q,\mathcal R})_{t,e}
  \in\mathsf T_{\mathcal R}\su(2)^m \enspace.
\]
Then $\mathcal T_q$ is the solution to the ordinary differential equation
\[
  \partial_q\mathcal T_q(\mathcal R)
  =\mathcal G(q,\mathcal T_q(\mathcal R)) \enspace,
  \qquad \mathcal T_0(\mathcal R)=\mathcal R  \]
for each $\mathcal R\in\su(2)^m$.
By \cref{clm:tv_vector_field_coefficient_bounds},
$\mathcal G$ is a $C^1$ time-dependent vector field on $\su(2)^m$.
Also,
\[
  \|\mathcal G(q,\mathcal R)\|_F
  =\left(\sum_{t,e}\|R_{t,e}B_{t,e}^{q,\mathcal R}\|_F^2\right)^{1/2}
  =\|\mathcal B^{q,\mathcal R}\|_F
  \le80kd\kappa^2\sqrt m\,\|X-Y\|_F\enspace.
\]
Thus, by \cite[Chapter~8, Section~1, Theorem~1.1 and Exercise~4]{Hirsch1976},
there is a unique $C^1$ diffeotopy
$(q,\mathcal R)\mapsto\mathcal T_q(\mathcal R)$ satisfying
\[
  \partial_q\mathcal T_q(\mathcal R)
  =\mathcal G(q,\mathcal T_q(\mathcal R)),
  \qquad \mathcal T_0(\mathcal R)=\mathcal R\enspace.
\]
Therefore, $\mathcal T_q$ is well defined for every $q\in[0,1]$,
and each $\mathcal T_q$ is a $C^1$ diffeomorphism of $\su(2)^m$.
This also proves \textup{(a)}.
We next prove \textup{(b)} and \textup{(c)}.

\paragraph{Proof of \textup{(b)}.}
Fix $\mathcal R\in\su(2)^m$.
By \cref{clm:tv_cutoff_derivative_bounds,clm:tv_vector_field_coefficient_bounds},
we have that for every $s\in[0,1]$,
\[
\begin{aligned}
  \left|\frac{\mathrm d}{\mathrm ds}f(\mathcal T_s(\mathcal R))\right|
  &\le\left\|(\mathrm d_{\mathcal R}f)(\mathcal T_s(\mathcal R))\right\|
       \|\mathcal B^{s,\mathcal T_s(\mathcal R)}\|_F\\
  &\le(8192k^3d^3\kappa^4m^{3/2})
       (80kd\kappa^2\sqrt m)\|X-Y\|_F\enspace.
\end{aligned}
\]
Integrating from $0$ to $q$ and using $q\le1$ gives
\[
\begin{aligned}
  |f(\mathcal T_q(\mathcal R))-f(\mathcal R)|
  &\le2^{20}k^4d^4\kappa^6m^2\|X-Y\|_F \le\frac{k^4\kappa^2}{4d^3T}
   \le2kd\kappa^2\enspace,
\end{aligned}
\]
where the second inequality uses
$\|X-Y\|_F\le(2^{20}\kappa^4d^9T^3)^{-1}$ and $m=dT/2$,
and the last uses $k\le d$ and $T\ge1$.
This proves \textup{(b)}.

\paragraph{Proof of \textup{(c)}.}
Fix $q\in[0,1]$.
Choose a Frobenius orthonormal basis $E_1,\ldots,E_{3m}$ of
$\sualg(2)^m$. For a real $C^1$ function $h$ on $\su(2)^m$
and $i\in [3m]$, write
\[
  D_i h(\mathcal R)
  \coloneqq\left.\frac{\mathrm d}{\mathrm d r}
    h(\mathcal R\exp(rE_i))\right|_{r=0},
  \qquad
  b_i(q,\mathcal R)\coloneqq
    \langle\mathcal B^{q,\mathcal R},E_i\rangle_F,
\]
where multiplication and exponentiation are componentwise.
For any $h$,
right invariance of $\nu_m$ gives
$\int h(\mathcal R\exp(rE_i))\,\mathrm d\nu_m(\mathcal R)
=\int h\,\mathrm d\nu_m$.
Differentiating at $r=0$ yields $\int D_i h\,\mathrm d\nu_m=0$.
Therefore, applying this to the product $b_i(q,\cdot)h$,
we have
\[
  0=\int D_i\bigl(b_i(q,\cdot)h\bigr)\,\mathrm d\nu_m
   =\int h\cdot D_i b_i(q,\cdot)\,\mathrm d\nu_m
    +\int b_i(q,\cdot) \cdot D_i h\,\mathrm d\nu_m\enspace.
\]
Hence,
\begin{equation}\label{eq:tv_integration_by_parts}
  \int b_i(q,\cdot) \cdot D_i h\,\mathrm d\nu_m
  =-\int h\cdot D_i b_i(q,\cdot)\,\mathrm d\nu_m\enspace.
\end{equation}

Fix a real $C^1$ function $h$ with $0\le h\le1$.
For $0\le s\le q$, define
$u_s=h\circ\mathcal T_q\circ\mathcal T_s^{-1}$.
Therefore, 
\begin{align*}
  0
  &=\frac{\mathrm d}{\mathrm ds}u_s(\mathcal T_s(\mathcal R))=(\partial_s u_s)(\mathcal T_s(\mathcal R))
    + (\mathrm d u_s)(\mathcal T_s(\mathcal R))
      [\partial_s\mathcal T_s(\mathcal R)]\\
  &=(\partial_s u_s)(\mathcal T_s(\mathcal R))
    +(\mathrm d u_s)(\mathcal T_s(\mathcal R))
      [\mathcal G(s,\mathcal T_s(\mathcal R))]\\
  &=(\partial_s u_s)(\mathcal T_s(\mathcal R))
    +\sum_{i=1}^{3m}b_i(s,\mathcal T_s(\mathcal R))
      D_i u_s(\mathcal T_s(\mathcal R))\enspace,
\end{align*}
where the last equality uses
$\mathcal G(s,\mathcal R)=\sum_{i=1}^{3m}b_i(s,\mathcal R)\mathcal R E_i$
and $D_i u_s(\mathcal R)=(\mathrm d u_s)(\mathcal R)[\mathcal R E_i]$.
Since $\mathcal T_s$ is surjective, this identity holds at every point
of $\su(2)^m$. Thus,
\[
  \partial_s u_s=-\sum_{i=1}^{3m}b_i(s,\cdot)D_i u_s\enspace.
\]
Therefore, we have
\begin{align*}
  \frac{\mathrm d}{\mathrm ds}\int u_s\,\mathrm d\nu_m
  &=\int\partial_s u_s\,\mathrm d\nu_m
   =-\sum_{i=1}^{3m}\int b_i(s,\cdot)D_i u_s\,\mathrm d\nu_m\\
  &=\sum_{i=1}^{3m}\int u_sD_i b_i(s,\cdot)\,\mathrm d\nu_m
   =\int u_s\sum_{i=1}^{3m}D_i b_i(s,\cdot)\,\mathrm d\nu_m\enspace,
\end{align*}
where the third equality uses \cref{eq:tv_integration_by_parts}.
By the definition of $b_i$ and $\|E_i\|_F=1$,
\[
  |D_i b_i(s,\mathcal R)|
  =\left|\left\langle
    \mathrm d_{\mathcal R}\mathcal B^{s,\mathcal R}[E_i],E_i
    \right\rangle_F\right|
  \le\|\mathrm d_{\mathcal R}\mathcal B^{s,\mathcal R}[E_i]\|_F
  \le\|\mathrm d_{\mathcal R}\mathcal B^{s,\mathcal R}\|_{\mathrm{op}}\enspace.
\]
Since $0\le h\le1$, the definition of $u_s$ gives $0\le u_s\le1$.
Thus, the triangle inequality and the fact that $\nu_m$ is a probability
measure imply
\begin{align*}
  \left|\frac{\mathrm d}{\mathrm ds}\int u_s\,\mathrm d\nu_m\right|
  &\le\int |u_s(\mathcal R)|\sum_{i=1}^{3m}|D_i b_i(s,\mathcal R)|
       \,\mathrm d\nu_m(\mathcal R)\\
  &\le\int\sum_{i=1}^{3m}|D_i b_i(s,\mathcal R)|
       \,\mathrm d\nu_m(\mathcal R)\\
  &\le3m\sup_{\mathcal R}
       \|\mathrm d_{\mathcal R}\mathcal B^{s,\mathcal R}\|_{\mathrm{op}}\\
  &\le3\cdot2^{19}k^3d^3\kappa^4m^3\|X-Y\|_F\enspace,
\end{align*}
where the last inequality uses \cref{clm:tv_vector_field_coefficient_bounds}.
Since $\mathcal T_0=\mathrm{id}$, the definition of $u_s$ gives
$u_0=h\circ\mathcal T_q$ and $u_q=h$.
We have
\begin{align*}
  \left|\int h\circ\mathcal T_q\,\mathrm d\nu_m
       -\int h\,\mathrm d\nu_m\right|
  &=\left|\int u_0\,\mathrm d\nu_m-\int u_q\,\mathrm d\nu_m\right|\\
  &=\left|-\int_0^q\frac{\mathrm d}{\mathrm ds}
       \left(\int u_s\,\mathrm d\nu_m\right)\,\mathrm ds\right|\\
  &\le\int_0^q\left|\frac{\mathrm d}{\mathrm ds}
       \int u_s\,\mathrm d\nu_m\right|\,\mathrm ds\\
  &\le\int_0^q3\cdot2^{19}k^3d^3\kappa^4m^3\|X-Y\|_F\,\mathrm ds\\
  &=3\cdot2^{19}qk^3d^3\kappa^4m^3\|X-Y\|_F\\
  &\le2^{18}\kappa^4d^9T^3\|X-Y\|_F\enspace,
\end{align*}
using $q\le1$, $m=dT/2$, and $k\le d$.
Taking the supremum over $h$ gives the total variation distance,
since every Borel indicator on this compact matrix group can be
approximated by $C^1$ functions taking values in $[0,1]$,
with convergence in $L^1$ under both measures.
Since $2^{18}\le C$, this proves \textup{(c)}.
\end{proof}

\section{\texorpdfstring{Proof of \cref{clm:tv_lift}}{Proof of Claim \ref*{clm:tv_lift}}}
\label[appendix]{app:tv_lift}

\tvLiftClaim*

\begin{proof}
Fix $X\in\stframe{d}{k}$ with $1\leq k<d$.
We first prove that $A_X(H)$ exists and is unique for every
$H\in\mathsf T_X\stframe{d}{k}$.
Extend the columns of $X$ to a unitary matrix $U=[X\;X_\perp]$, and write
\[
    U^\dagger H=\begin{pmatrix}K\\L\end{pmatrix},
    \qquad K=X^\dagger H,\quad L=X_\perp^\dagger H\enspace.
\]
Since $H$ is tangent at $X$, we have
$K+K^\dagger=X^\dagger H+H^\dagger X=0$.
For any $A\in\sualg(d)$ satisfying $AX=H$, let
$\widetilde A=U^\dagger AU$.
Since $U^\dagger X=\binom{I_k}{0}$, we have
\[
    \widetilde A\begin{pmatrix}I_k\\0\end{pmatrix}
    =U^\dagger AX=U^\dagger H
    =\begin{pmatrix}K\\L\end{pmatrix}\enspace.
\]
Thus, the first $k$ columns of $\widetilde A$ are fixed.
Also, $\widetilde A^\dagger=-\widetilde A$ and
$\tr{\widetilde A}=\tr{A}=0$, so
\begin{equation}
    U^\dagger AU=
    \begin{pmatrix}
        K&-L^\dagger\\
        L&C
    \end{pmatrix},
    \qquad C^\dagger=-C,\quad \tr{C}=-\tr{K}\enspace.
    \label{eq:tv_lift_block_constraints}
\end{equation}
Conversely, for any $C$ satisfying these two conditions,
it is not hard to check that the matrix
$A$ given by this block representation belongs to $\sualg(d)$ and
satisfies $AX=H$.
Therefore, minimizing $\norm{A}_F$ over all such $A$ reduces to
minimizing $\norm{C}_F$ subject to the conditions in
\cref{eq:tv_lift_block_constraints}.
By the Cauchy--Schwarz inequality,
\[
    |\tr{C}|^2
    =\left|\sum_{i=1}^{d-k}C_{ii}\right|^2
    \leq(d-k)\sum_{i=1}^{d-k}|C_{ii}|^2
    \leq(d-k)\norm{C}_F^2\enspace.
\]
Using the trace constraint in \cref{eq:tv_lift_block_constraints},
we obtain
\[
    \norm{C}_F^2\geq\frac{|\tr{C}|^2}{d-k}
    =\frac{|\tr{K}|^2}{d-k}\enspace.
\]
The equality holds if and only if all diagonal
entries of $C$ are equal and all off-diagonal entries are zero.
The trace constraint therefore gives the unique equality case
$C=-\tr{K}I_{d-k}/(d-k)$.
Since $\norm{A}_F^2=\norm{K}_F^2+2\norm{L}_F^2+\norm{C}_F^2$,
the minimizer is
\begin{equation}
    A_X(H)=U
    \begin{pmatrix}
        K&-L^\dagger\\
        L&-\dfrac{\tr{K}}{d-k}I_{d-k}
    \end{pmatrix}U^\dagger\enspace.
    \label{eq:tv_lift_block_formula}
\end{equation}
Expanding the product in \cref{eq:tv_lift_block_formula}, we obtain
\begin{align*}
    A_X(H)
    &=XKX^\dagger-XL^\dagger X_\perp^\dagger
        +X_\perp LX^\dagger
        -\frac{\tr{K}}{d-k}X_\perp X_\perp^\dagger\\
    &=X(X^\dagger H)X^\dagger-XH^\dagger(I_d-XX^\dagger)
        +(I_d-XX^\dagger)HX^\dagger -\frac{\tr{X^\dagger H}}{d-k}(I_d-XX^\dagger)\enspace,
\end{align*}
where we used $K=X^\dagger H$, $L=X_\perp^\dagger H$, and
$X_\perp X_\perp^\dagger=I_d-XX^\dagger$.
Since $H^\dagger X=-X^\dagger H$, this simplifies to
\begin{equation}
    A_X(H)=HX^\dagger-XH^\dagger-X(X^\dagger H)X^\dagger
        -\frac{\tr{X^\dagger H}}{d-k}(I_d-XX^\dagger)\enspace.
    \label{eq:tv_lift_explicit_formula}
\end{equation}
The expression in \cref{eq:tv_lift_explicit_formula} depends only on
$X$ and $H$, so it is independent of the choice of $X_\perp$.
It also shows that $A_X(H)$ depends linearly on $H$ for fixed $X$.

We now verify the three properties in \cref{clm:tv_lift}.

\vspace{-0.8em}
\paragraph{Proof of the first property.}
Since $U$ is unitary, $\norm{H}_F^2=\norm{K}_F^2+\norm{L}_F^2$.
Thus, \cref{eq:tv_lift_block_formula} gives
\begin{align}
    \norm{A_X(H)}_F^2
    &=\norm{K}_F^2+2\norm{L}_F^2
        +\frac{|\tr{K}|^2}{d-k} =2\norm{H}_F^2-\norm{X^\dagger H}_F^2
        +\frac{|\tr{X^\dagger H}|^2}{d-k}\enspace.
    \label{eq:tv_lift_norm}
\end{align}
Since $\norm{X^\dagger H}_F\leq\norm{H}_F$,
\cref{eq:tv_lift_norm} gives
\[
    \norm{A_X(H)}_F^2
    \geq2\norm{H}_F^2-\norm{X^\dagger H}_F^2
    \geq\norm{H}_F^2\enspace.
\]
For the upper bound, the Cauchy--Schwarz inequality gives
\[
    \abs{\tr{X^\dagger H}}^2
    \leq k\norm{X^\dagger H}_F^2
    \leq k\norm{H}_F^2\enspace.
\]
Substituting this estimate into \cref{eq:tv_lift_norm}, we obtain
\[
    \norm{A_X(H)}_F^2
    \leq\left(2+\frac{k}{d-k}\right)\norm{H}_F^2
    \leq(d+1)\norm{H}_F^2
    \leq2d\norm{H}_F^2\enspace,
\]
where we used $d-k\geq1$, $k\leq d-1$, and $d\geq2$.

\paragraph{Proof of the second property.}
Let $G\in\mathbb C^{d\times d}$ be unitary.
Applying \cref{eq:tv_lift_explicit_formula} to $GX$ and $GH$, and using
$(GX)^\dagger(GH)=X^\dagger H$ and
$I_d-GXX^\dagger G^\dagger=G(I_d-XX^\dagger)G^\dagger$, we have
\begin{align*}
    A_{GX}(GH)
    &=GHX^\dagger G^\dagger-GXH^\dagger G^\dagger
        -GX(X^\dagger H)X^\dagger G^\dagger -\frac{\tr{X^\dagger H}}{d-k}
        G(I_d-XX^\dagger)G^\dagger\\
    &=GA_X(H)G^\dagger\enspace.
\end{align*}
Therefore, unitary invariance of the Frobenius norm gives
\[
    \norm{A_{GX}(GH)}_F
    =\norm{GA_X(H)G^\dagger}_F
    =\norm{A_X(H)}_F\enspace.
\]

\paragraph{Proof of the third property.}
Suppose every row of $X$ has squared Frobenius norm at most $\Delta>0$,
and $B$ has blocks $B_e\in\sualg(2)$ on a perfect matching $M$.

With $U=[X\;X_\perp]$ as above, if a matrix $C\in\sualg(d)$ satisfies $CX=0$, then
$U^\dagger CU=\operatorname{diag}(0,C_0)$ with $\tr{C_0}=0$.
By \cref{eq:tv_lift_block_formula}, we have $\langle A_X(H),C\rangle_F = 0$.
Since $A_X(BX)X=BX$, the matrix $C=B-A_X(BX)\in\sualg(d)$
satisfies $CX=0$.
The preceding orthogonality therefore gives
\[
    0=\langle A_X(H),B-A_X(BX)\rangle_F \enspace.
\]
Thus, we have
\[
	\langle A_X(H),A_X(BX)\rangle_F
    =\langle A_X(H),B\rangle_F\enspace.
\]
Expanding the squared norm, we obtain
\begin{align}
    \norm{A_X(H+BX)}_F^2
    &=\norm{A_X(H)+A_X(BX)}_F^2\nonumber\\
    &=\norm{A_X(H)}_F^2
        +2\langle A_X(H),A_X(BX)\rangle_F
        +\norm{A_X(BX)}_F^2\nonumber\\
    &=\norm{A_X(H)}_F^2+2\langle A_X(H),B\rangle_F
        +\norm{A_X(BX)}_F^2\enspace.
    \label{eq:tv_lift_exact_change}
\end{align}
For each $e=(i,j)\in M$, let $X_e$ consist of rows $i,j$ of $X$.
Since $\norm{X_e}_F^2\leq2\Delta$,
\begin{equation}
    \norm{BX}_F^2
    =\sum_{e\in M}\norm{B_eX_e}_F^2
    \leq\sum_{e\in M}\norm{B_e}_F^2\norm{X_e}_F^2
    \leq2\Delta\norm{B}_F^2\enspace.
    \label{eq:tv_lift_matching_action_bound}
\end{equation}
Also, $\tr{B}=0$ and
$\norm{XX^\dagger-(k/d)I_d}_F^2=k(d-k)/d$.
Thus, by the Cauchy--Schwarz inequality,
\begin{equation}
    \frac{|\tr{X^\dagger BX}|^2}{d-k}
    =\frac{|\tr{B(XX^\dagger-(k/d)I_d)}|^2}{d-k}
    \leq\frac kd\norm{B}_F^2\enspace.
    \label{eq:tv_lift_trace_bound}
\end{equation}
Applying \cref{eq:tv_lift_norm} to $BX$ and using
\cref{eq:tv_lift_matching_action_bound,eq:tv_lift_trace_bound}, we have
\begin{align*}
    \norm{A_X(BX)}_F^2
    &\leq2\norm{BX}_F^2+\frac{|\tr{X^\dagger BX}|^2}{d-k} \leq\left(4\Delta+\frac kd\right)\norm{B}_F^2
    \leq5\Delta\norm{B}_F^2\enspace,
\end{align*}
where the last inequality uses $k=\norm{X}_F^2\leq d\Delta$.
Substituting this bound into \cref{eq:tv_lift_exact_change}
proves \cref{eq:tv_lift_change}.
\end{proof}

\endgroup

\end{document}